\documentclass[a4paper,USenglish, thm-restate]{lipics-v2021}
\usepackage{xcolor}

\pdfoutput=1 
\hideLIPIcs  

\renewenvironment{definition}{\begin{olddefinition}\rm}
                              {\end{olddefinition}}

\title{Solving Partial Dominating Set and Related Problems Using Twin-Width}

\author{Jakub Balab\' an}{Faculty of Informatics, Masaryk University, Brno, Czechia}{jakbal@mail.muni.cz}{https://orcid.org/0000-0002-2475-8938}{Brno Ph.D. Talent Scholarship Holder – Funded by the Brno City Municipality}

\author{Daniel Mock}{Dept.\ of Computer Science, RWTH Aachen
University, Aachen, Germany}{mock@cs.rwth-aachen.de}{https://orcid.org/0000-0002-0011-6754}{}

\author{Peter Rossmanith}{Dept.\ of Computer Science, RWTH Aachen
University, Aachen, Germany}{rossmani@cs.rwth-aachen.de}{https://orcid.org/0000-0003-0177-8028}{}

\funding{\emph{Daniel Mock and Peter Rossmanith:} Supported by the Deutsche Forschungsgemeinschaft (DFG, German Science Foundation) -- DFG-927/15-2}

\authorrunning{Jakub Balab\' an, Daniel Mock, Peter Rossmanith} 

\Copyright{Jakub Balab\' an, Daniel Mock, Peter Rossmanith} 



\acknowledgements{We would like to thank one of the anonymous reviewers for suggesting a generalization of the logic we considered that also captures the \textsc{Partial Vertex Cover} problem.}

\nolinenumbers 

\EventEditors{Pawe\l{} Gawrychowski, Filip Mazowiecki, and Micha\l{} Skrzypczak}
\EventNoEds{3}
\EventLongTitle{50th International Symposium on Mathematical Foundations of Computer Science (MFCS 2025)}
\EventShortTitle{MFCS 2025}
\EventAcronym{MFCS}
\EventYear{2025}
\EventDate{August 25--29, 2025}
\EventLocation{Warsaw, Poland}
\EventLogo{}
\SeriesVolume{345}
\ArticleNo{62}

\usepackage{tikz}
\usetikzlibrary{positioning, shapes.geometric}
\usetikzlibrary{arrows}
\usetikzlibrary{fit}
\usepackage[ruled, vlined, linesnumbered]{algorithm2e}
\usepackage{todonotes} 

\newcommand{\sep}{\colon}
\newcommand{\seq}{\subseteq}
\newcommand{\bb}{\mathbb}
\newcommand{\ca}{\mathcal}
\newcommand{\N}{\mathbb{N}}

\DeclareMathOperator{\range}{range}

\newcommand{\decisionproblem}[4]{
\medskip
    \noindent\textsc{#1} \\
    \noindent\textit{Input:} #2 \\
    \noindent\textit{Question:} #3\\
    \noindent\textit{Parameter:} #4
\smallskip
}

\definecolor{Tolpurple}{HTML}{AA3377}

\begin{document}

\maketitle

\begin{abstract}
Partial vertex cover and partial dominating set are two well-investigated
optimization problems. While they are $\rm W[1]$-hard on general graphs,
they have been shown to be fixed-parameter tractable on many sparse graph
classes, including nowhere-dense classes. In this paper, we demonstrate
that these problems are also fixed-parameter tractable with respect to
the twin-width of a graph. Indeed, we establish a more general result:
every graph property that can be expressed by a logical formula of the form
$\phi\equiv\exists x_1\cdots \exists x_k \sum_{\alpha \in I} \#y\,\psi_\alpha(x_1,\ldots,x_k,y)\ge t$,
where $\psi_\alpha$ is a quantifier-free formula for each $\alpha \in I$,
$t$ is an arbitrary number,
and $\#y$ is a counting quantifier, can be evaluated in time $f(d,k)n$,
where $n$ is the number of vertices and $d$ is the width of a contraction
sequence that is part of the input. 
In addition to the 
aforementioned problems, this includes also connected partial dominating set and
 independent partial dominating set.
\end{abstract}

\section{Introduction}\label{sec:intro}

{\sc Vertex Cover}, {\sc Independent Set}, and {\sc Dominating
Set} are well-known graph problems that have driven significant
advancements in parameterized complexity theory. Despite being
NP-complete, these problems motivated the development of new
approaches, as the difficulty of finding solutions of size $k$ appears
to increase significantly from one problem to the
next~\cite{DowneyF1999}. On general graphs, {\sc Vertex Cover} has
been
found to be fixed-parameter tractable, while {\sc Independent Set} is
$\rm W[1]$-complete and {\sc Dominating Set} is $\rm W[2]$-complete.
Consequently, researchers have designed faster algorithms for solving
{\sc Vertex Cover} on general graphs~\cite{ChenKX2010,
BalasubramanianFR1998, ChandranG2005, ChenKJ2001, ChenLJ2000,
DowneyF1992, NiedermeierR1999, NiedermeierR2003, StegeRHH2002}, while
identifying more and more general graph classes that admit fixed-parameter
algorithms for the other two problems. This progression started with
planar graphs~\cite{AlberBFKN2002} and led to nowhere-dense graph classes
via intermediate results~\cite{DemaineFHT2005, EllisFF2004, PhilipRS2009,
DawarK2009}.  All three problems can be expressed in first-order logic,
and significant advances have been made in meta-algorithm design
to solve the first-order model-checking problem. This progression
begins with bounded-degree graphs, continues through bounded genus,
minor-closed, bounded expansion, nowhere-dense, and finally monadically
stable classes~\cite{Seese1996, FrickG2001, FlumG2001, DawarGK2007,
GroheKS2014, DreierEMMPT2024}.

Twin-width is a recently introduced width measure that generalizes
clique-width~\cite{BonnetKTW2020}.
Classes of bounded twin-width do not contain
bounded degree graph classes and are not necessarily monadically
stable. They are therefore incomparable to the sparsity hierarchy.
Bonnet et al. have developed a parameterized algorithm for first-order
model-checking on classes of bounded twin-width~\cite{BonnetKTW2020}.
Please note that, unlike tree-width, there is no efficient way known
to compute or even approximate twin-width. Moreover, most twin-width
based algorithms, including first-order model-checking, require as
part of the input not only the graph but also a contraction sequence.
The contraction sequence proves that the graph has small twin-width and
guides an algorithm. This guidance can be compared to that of dynamic
programming on tree decompositions. An optimal tree decomposition
can, however, be relatively efficiently computed and even faster
approximated. In the world of sparsity algorithms operating in classes
of bounded expansion or nowhere-dense classes, the situation is even
better: no underlying structure is needed; the algorithm works correctly
on every graph but becomes slow when applied outside its scope.

While this represents a significant body of work on {\sc Vertex
Cover}, {\sc Independent Set}, {\sc Dominating Set}, and many other
first-order expressible
problems, the story does not end here. In {\sc Vertex Cover}, you
ask for $k$
nodes that cover all edges, and in {\sc Dominating Set}, you seek $k$
nodes
that dominate all other nodes. Generalizations of these problems are the
partial versions: given $k$ and $t$, are there $k$ nodes that cover $t$
edges or are there $k$ nodes that dominate $t$ nodes? These problems
are known as {\sc Partial Vertex Cover} and {\sc Partial Dominating
Set}; note that in the parameterized setting, $k$ is the parameter. 
As generalizations, they turn out to be harder than their complete
counterparts. For example, while {\sc Vertex Cover} is fixed-parameter
tractable on general graphs, {\sc Partial Vertex Cover} has been
shown to
be $\rm W[1]$-hard~\cite{JiongNW2007}, and while {\sc Dominating Set}
is
fixed-parameter tractable on degenerate graphs, {\sc Partial
Dominating Set}
is $\rm W[1]$-hard~\cite{GolovachV2008}.
However, it is fixed-parameter tractable with
respect to~$t$ even on general graphs~\cite{KneisMR2007}.
Another example of how the partial variant can be harder are
$c$-closed graphs, where non-adjacent vertices share at most $c$
neighbors.  {\sc Dominating Set} is fixed-parameter tractable for
bounded $c$~\cite{DBLP:journals/siamdm/KoanaKS22}, while {\sc Partial Dominating Set} becomes a hard problem
even for~$c=2$~\cite{KaneshMRSS23}.

Another sequence of papers has demonstrated that \textsc{Partial Dominating
Set} can be solved in time $f(k)n^{O(1)}$ for larger and larger graph
classes. Amini, Fomin, and Saurabh developed an algorithm with a
running time of $O(f(k) \cdot t \cdot n^{C_H})$ for the class of
$H$-minor-free graphs, where $C_H$ is a constant dependent only
on~$H$. The graph classes with bounded expansion and nowhere-dense
graph classes introduced in the sparsity project by Nešetřil and
Ossona de Mendez~\cite{NesetrilM2012} contain all $H$-minor-free classes.
For bounded expansion, there exists a meta-algorithm that solves many
problems, including \textsc{Partial Dominating Set} in linear fixed-parameter
tractable time, i.e., in $f(k)n$ steps~\cite{DreierR2021}. Additionally,
for nowhere-dense classes, there is a more restricted meta-algorithm
that is general enough to solve \textsc{Partial Dominating Set} in time $O(f(k)
n^{1+\epsilon})$ for every $\epsilon > 0$~\cite{DreierMR2023}.

The ``simpler'' {\sc Partial Vertex Cover} has not seen as much
attention, but has still been investigated as well; see, for
example~\cite{MkrtchyanP2022,FominGIK2024,MkrtchyanPSW2024}.

\subparagraph*{Our contribution.}

We demonstrate that both {\sc Partial Dominating Set} and {\sc Partial
Vertex Cover} can be solved efficiently on graph classes with bounded
twin-width. To achieve this, we develop two dynamic programming algorithms
that work along a contraction sequence provided as part of the input.
Interestingly, both algorithms share similarities in their complexity,
while usually {\sc Partial Vertex Cover} is the much easier problem.
The algorithms are based on ideas developed
by Bonnet et\ al.\ in order to solve the (complete) {\sc Dominating Set}
problem on bounded twin-width~\cite{tww3}.

Having separate algorithms for related problems is not ideal and it
would be beneficial to have them emerge as special cases within a more
comprehensive framework. Again, the language of logic could facilitate
this. We present a meta-algorithm that solves the model-checking
problem for formulas of the form $\phi \equiv \exists x_1\cdots\exists x_k\# y\,
\psi(x_1,\ldots,x_k,y)\geq t$ where $\psi$ must be quantifier-free (this logic
was considered in~\cite{DreierMR2023}).
Such a formula is true in a graph $G=(V,E)$ if it contains nodes
$v_1,\ldots,v_k$ such that $\psi(v_1,\ldots,v_k,u)$ holds for at least $t$
many nodes $u\in V$. For instance, with
$\psi(x_1,\ldots,x_k,y)=\bigvee_{i=1}^k(E(x_i,y)\lor
x_i=y)$, we can solve {\sc Partial Dominating Set}.
With a quantifier-free
$\psi$, we can express numerous other intriguing problems, including
{\sc Connected Partial Dominating Set} and {\sc Independent Partial
Dominating Set}, 
as well as more exotic ones like ``Are there $k$ nodes that are all
part of triangles and dominate at least $t$ of all other nodes?'' or
``How many vertices do we have to delete in order to get a $2$-independent
set of size~$k$?'' (see \cite{HotaPP2001,AttiyaW2004} for applications).
The time needed to decide $G\models\phi$ is $d^{O(k^2d)}n$, which is
linear FPT time.  ``Linear'' means that the input has to be shorter
than the classical encoding of the underlying graph as an adjacency
list or adjacency matrix.  While graphs of bounded twin-width can have
a quadratic number of edges, they can still be encoded via their
much shorter contraction sequence \cite{BonnetKTW2020, DBLP:conf/icalp/GajarskyPPT22}.
Note that our single-purpose algorithm for {\sc Partial Dominating Set}
achieves a better running time of $(kd)^{\ca O(kd)}n$.

Our meta-algorithm can, in fact, model-check even more general formulas,
namely formulas of the form $\phi \equiv \exists x_1\cdots\exists x_k \sum_{\alpha \in I} \# y\,
\psi_\alpha(x_1,\ldots,x_k,y)\geq t$, where $\psi_\alpha$ is a quantifier-free formula
for each $\alpha \in I$.
Such a formula $\phi$ is true if the \emph{sum} of the numbers of vertices $y$
satisfying the formulas $\psi_\alpha$ is at least $t$.
Using such a formula, we can also express the \textsc{Partial Vertex Cover} problem,
see Example~\ref{example:pvc}.

Finally, let us remark that if we take a power of a graph,
then its twin-width remains small~\cite{BonnetKTW2020}.
This simple observation allows us to tackle the ``distance-$r$''
variants of {\sc Partial Dominating Set} as well by using transductions.

\section{Preliminaries}\label{sec:prelims}
For integers $i$ and $j$, we let $[i,j] := \{n \in \bb N \sep i \le n \le j\}$ and $[i] := [1, i]$.
We study finite, simple and undirected graphs with non-empty vertex set. 

\subparagraph*{Twin-width}

A \emph{trigraph} $G$ is a graph whose edge set is partitioned into a set of \emph{black} and \emph{red} edges. The set of red edges is denoted $R(G)$, and the set of black edges $B(G)$;
standard graphs are viewed as trigraphs with no red edges.
The \emph{red graph} $G^R$ of $G$ is the graph $(V(G), R(G))$
and the \emph{black graph} $G^B$ of $G$ is the graph $(V(G), B(G))$.
Any graph-theoretic notion with a color adjective (black or red) has its standard meaning applied in the corresponding graph.
For example, the \emph{red degree} of $u\in V(G)$ in $G$ is its degree in $G^R$.

Given a trigraph $G$, a \emph{contraction} of two distinct vertices $u,v\in V(G)$ is the operation which produces a new trigraph by (1) removing $u, v$ and adding a new vertex $w$, (2) adding a black edge $wx$ for each $x\in V(G)$ such that $xu$, $xv\in B(G)$, and (3) adding a red edge $wy$ for each $y\in V(G)$ such that $yu\in R(G)$, or $yv\in R(G)$, or $y$ contains only a single black edge to either $v$ or $u$.
A sequence $C = (G = G_1,\ldots,G_n)$ is a \emph{contraction sequence of $G$} if it is a sequence of trigraphs such that $|V(G_n)| = 1$ and for all $i\in [n-1]$, $G_{i+1}$ is obtained from $G_i$ by contracting two vertices.
The \emph{width} of a contraction sequence $C$ is the maximum red degree over all vertices in all trigraphs in $C$.
The \emph{twin-width} of $G$
is the minimum width of any contraction sequence of $G$.
An example of a contraction sequence is provided in Figure~\ref{fig:seq}.

\begin{figure}
\begin{tikzpicture}[line cap=round,line join=round,>=triangle 45,x=1.0cm,y=1.0cm]

\tikzset{
    vertex/.style = {draw, circle, fill=gray, minimum width=4pt, inner sep=0pt}}
    
\begin{scriptsize}
\node[vertex] (a) {};
\node[vertex] (b) [right of = a] {};
\node[vertex] (c) [above of = a] {};
\node[vertex] (d) [right of = c] {};
\node[vertex] (e) [above of = c] {};
\node[vertex] (f) [right of = e] {};
\draw (a)--(b)--(c)--(f)--(e)--(d)--(b);
\draw(a)--(c)--(e);

\node () at (a) [left=2pt] {$A$};
\node () at (c) [left=2pt] {$C$};
\node () at (e) [left=2pt] {$E$};
\node () at (b) [right=2pt] {$B$};
\node () at (d) [right=2pt] {$D$};
\node () at (f) [right=2pt] {$F$};

\node[vertex] (a2) [right = 40pt of b]{};
\node[vertex] (b2) [right of = a2] {};
\node[vertex] (c2) [above of = a2] {};
\node[vertex] (d2) [right of = c2] {};
\node[vertex] (ef) [above of = c2] {};
\draw (a2)--(b2)--(c2)--(ef)--(d2)--(b2);
\draw(a2)--(c2);
\draw[color=red, thick] (d2)--(ef);

\node () at (a2) [left=2pt] {$A$};
\node () at (c2) [left=2pt] {$C$};
\node () at (b2) [right=2pt] {$B$};
\node () at (d2) [right=2pt] {$D$};
\node () at (ef) [right=2pt] {$EF$};

\node[vertex] (ab) [right = 40pt of b2]{};
\node[vertex] (c3) [above of = ab] {};
\node[vertex] (d3) [right of = c3] {};
\node[vertex] (ef3) [above of = c3] {};
\draw (ab)--(c3)--(ef3);
\draw[color=red, thick] (ab)--(d3)--(ef3);

\node () at (ab) [right=2pt] {$AB$};
\node () at (c3) [left=2pt] {$C$};
\node () at (d3) [left=3pt] {$D$};
\node () at (ef3) [right=2pt] {$EF$};

\node[vertex] (cd) [right = 40pt of d3] {};
\node[vertex] (ab4) [below of= cd]{};
\node[vertex] (ef4) [above of = cd] {};
\draw[color=red, thick] (ab4)--(cd)--(ef4);
\node () at (ab4) [left=2pt] {$AB$};
\node () at (cd) [left=2pt] {$CD$};
\node () at (ef4) [left=2pt] {$EF$};

\node[vertex] (cdef) [right = 40pt of cd] {};
\node[vertex] (ab5) [below of= cdef]{};
\draw[color=red, thick] (ab5)--(cdef);
\node () at (ab5) [left=2pt] {$AB$};
\node () at (cdef) [above = 2pt] {$CDEF$};

\node[vertex] (fin) [right = 40pt of ab5] {};
\node () at (fin) [above = 2pt] {$ABCDEF$};

\end{scriptsize}
\end{tikzpicture}
\caption{A contraction sequence of width 2 for the leftmost graph, consisting of $6$ trigraphs.
\label{fig:seq}}
\end{figure}
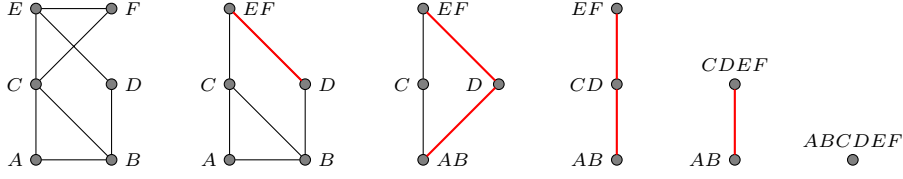

If a contraction sequence $C = (G_1,\ldots,G_n)$ of $G$ is clear from context, then for each $i \in [n]$ and $u \in V(G_i)$, we call the set of vertices of $G$ contracted to $u$ the \emph{bag} of $u$ and we denote it $\beta(u)$. Formally, $\beta(u) = \{u\}$ if $u \in V(G)$, and $\beta(u) = \beta(v) \cup \beta(w)$ if there is $i \in [2,n]$ such that $u \in V(G_i)$, $v, w \in V(G_{i-1})$, and $u$ is created by contracting $v$ and $w$.
Note that if a vertex appears in multiple trigraphs in $C$, then its bag is the same in all of them.
If $T \seq V(G_i)$ for some $i \in [n]$, then we use $\beta(T)$ to denote $\bigcup_{u \in T} \beta(T)$.

\subparagraph*{Basic profiles}

Now we define \emph{basic profiles}, inspired by the $k$-profiles in the \textsc{Dominating Set} algorithm from~\cite{tww3}.
Intuitively, a profile describes a ``type'' of a subsolution we compute during our dynamic algorithms.
For each profile, we will compute a subsolution that is, in some sense, best for the profile.
In the end, we can read the answer to our problem from the profiles associated with the
last trigraph $G_n$ in the contraction sequence.

In the following sections, we will expand the basic profiles with extra information in two different ways:
to \emph{extended profiles}, which we will use to solve \textsc{Partial Dominating Set}
and \textsc{Partial Vertex Cover} in Section~\ref{sec:problems}, and to \emph{virtual profiles},
which we will use for the model-checking in Section~\ref{sec:logic}.
Note that the $k$-profiles in~\cite{tww3} can also be viewed as extensions of our basic profiles.
After defining the basic profiles, we will prove some of their properties,
which we will later generalize to the extended and virtual profiles.

Let us fix a graph $G$, a contraction sequence $(G = G_1 ,\ldots,G_n)$ of width $d$, and an integer $k \in \bb N$. 

\begin{definition}
If $i \in [n]$, then a \emph{basic $i$-profile} is a pair $(T, D)$ such that $T \seq V(G_i)$, $|T| \le k(d+1)$, $T$ is connected in $G_i^R$, $D \seq T$, and $|D| \le k$.
\end{definition}

For example, in the \textsc{Partial Vertex Cover} algorithm, $k$ will be the number of covering vertices, $D$ will be the set such that the covering vertices are in $\beta(D)$, and $T$ is the ``context'' in which the cover is considered, i.e., we are interested in which edges of $G[\beta(T)]$ are covered (this intuition will be formalized in Definition~\ref{def:sol}).

Now we define how a basic $i$-profile can be decomposed into several basic $(i-1)$-profiles, again inspired by~\cite{tww3}.
Later, we will analogously define decompositions for extended and virtual profiles. 

Let $i \in [2, n]$, let $v \in V(G_i)$ be the new vertex in $G_i$, let $u_1, u_2 \in V(G_{i-1})$ be the two vertices contracted into $v$, let $\pi = (T, D)$ be a basic $i$-profile such that $v \in T$, and let $T' = (T\setminus v)\cup \{u_1,u_2\}$.
We say that a set $D' \seq T'$ is \emph{compatible} with $\pi$ if $D \setminus \{v\} = D' \setminus \{u_1, u_2\}$, $|D'| \le k$, and $v \in D$ if and only if $u_1 \in D'$ or $u_2 \in D'$.
Note that the condition $|D'| \le k$ is necessary only to prohibit the case $|D| = k$ and $u_1, u_2 \in D'$.

We say that a set $\{(T_1, D_1), \ldots, (T_m, D_m)\}$ of basic $(i-1)$-profiles is a \emph{$D'$-decomposition} of $\pi$, see Figure~\ref{fig:basic}, if:
\begin{enumerate}
\item The sets $T_1, \ldots, T_m$ are pairwise disjoint.
\item $T' = \bigcup_{j \in [m]} T_j$, $D' = \bigcup_{j \in [m]} D_j$, and $D'$ is compatible with $\pi$.\label{decomp-t}
\item If $x \in T_j$ and $y \in D_\ell$ for $j \ne \ell$, then $xy \notin R(G_{i-1})$.\label{decomp-red}
\end{enumerate}

\begin{figure}
    \hspace{50pt}
    \centering
    \begin{subfigure}[t]{0.25\textwidth}
        \centering
\begin{tikzpicture}[line cap=round,line join=round,>=triangle 45,x=1.0cm,y=1.0cm]
\begin{scope}[every node/.style={draw, circle, minimum width=8pt, inner sep=0pt}]
\node (aa)[fill=pink, line width= 2pt] {};
\node (ab) [right= of aa, fill=blue!60!white] {};
\node (ac) [right= of ab, fill=blue!60!white] {};
\node (ad) [right= of ac] {};

\node (ba) [below = of aa] {};
\node (bb) [right= of ba, fill=blue!60!white, line width= 2pt] {};
\node (bc) [right= of bb, fill=green] {};
\node (bd) [right= of bc, fill=green, line width= 2pt] {};

\node (ca) [below = of ba, fill=yellow, line width= 2pt] {};
\node (cb) [right= of ca, fill=yellow] {};
\node (cc) [right= of cb] {};
\node (cd) [right= of cc, fill=green, line width= 2pt] {};

\draw[color=red, thick] (cd)--(bd)--(ad)--(ac)--(ab)--(bb);
\draw[color=red, thick] (aa)--(ba)--(ca)--(cb);
\draw (ab)--(aa)--(cb)--(cc)--(bb)--(bc);
\draw (bc)--(cc)--(cd);
\draw (bd)--(ac);
\draw[color=red, thick] (bc)--(bd);
\end{scope}

\node[draw,rounded corners,dashed,fit=(bb) (cb)] {} ;
\end{tikzpicture}    \end{subfigure}
    \hfill
    \begin{subfigure}[t]{0.25\textwidth}
        \centering
\begin{tikzpicture}[line cap=round,line join=round,>=triangle 45,x=1.0cm,y=1.0cm]
\begin{scope}[every node/.style={draw, circle, minimum width=8pt, inner sep=0pt}]
\node (aa)[fill=gray!60!white, line width= 2pt] {};
\node (ab) [right= of aa, fill=gray!60!white] {};
\node (ac) [right= of ab, fill=gray!60!white] {};
\node (ad) [right= of ac] {};

\node (ba) [below = of aa] {};
\node (bb) [right= of ba, opacity=0] {};
\node (bc) [right= of bb, fill=gray!60!white] {};
\node (bd) [right= of bc, fill=gray!60!white, line width= 2pt] {};

\node (ca) [below = of ba, fill=gray!60!white, line width= 2pt] {};
\node (cb) [right= of ca, opacity=0] {};
\node (cc) [right= of cb] {};
\node (cd) [right= of cc, fill=gray!60!white, line width= 2pt] {};

\path (cb) -- (bb) coordinate[midway] (midpoint);
\node (new) [minimum width=12pt, line width= 2pt, fill=gray!60!white] at (midpoint) {};

\draw[color=red, thick] (cd)--(bd)--(ad)--(ac)--(ab)--(new);
\draw[color=red, thick] (aa)--(ba)--(ca)--(new);
\draw (ab)--(aa);

\draw (cc)--(new);
\draw (bc)--(cc)--(cd);
\draw (bd)--(ac);
\draw[color=red, thick] (aa)--(new)--(bc)--(bd);
\end{scope}
\node[draw,rounded corners,dashed,fit=(bb) (cb), opacity= 0] {} ;
\end{tikzpicture}    \end{subfigure}
    \hspace{50pt}
    \caption{\textbf{Right:} A basic $i$-profile $\pi = (T, D)$. Vertices of $T$ are colored in gray, and vertices of $D$ have thick boundary. The new vertex of $G_i$ is represented by the bigger disc. \textbf{Left:} A $D'$-decomposition $\{(T_1, D_1), \ldots, (T_4, D_4)\}$ of $\pi$. The sets $T_1, \ldots, T_4$ are represented by the colors blue, green, yellow, and pink. Vertices of $D'$ have thick boundary, and the two vertices to be contracted are inside a dashed rectangle.}
    \label{fig:basic}
\end{figure}
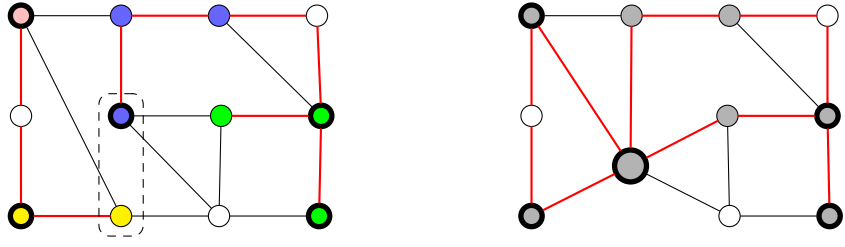

Let us give some intuition behind condition~\ref{decomp-red}.
If $y \in D_\ell$, then there is a dominating (resp. covering) vertex $y' \in \beta(y)$.
If $xy \in R(G_{i-1})$, then we would not know which vertices of $\beta(x)$ are dominated by $y'$ (resp. how many edges in $\{x'y' \sep x'\in\beta(x)\}$ are covered by $y'$).
Hence, the connections between profiles must be homogeneous (either a black edge or a non-edge).

Now let us prove that a small decomposition always exists (and can be found efficiently).

\begin{lemma}\label{lem:basic-decomposition}
If $i \in [2, n]$, $V(G_i) \setminus V(G_{i-1}) = \{v\}$, $V(G_{i-1}) \setminus V(G_i) = \{u_1, u_2\}$, $\pi = (T, D)$ is a basic $i$-profile such that $v \in T$, $T' = (T\setminus v)\cup \{u_1,u_2\}$, and $D' \seq T'$ is compatible with $\pi$, then there is a $D'$-decomposition $\ca D$ of $\pi$ containing at most $d+2$ basic $(i-1)$-profiles.
Moreover, $\ca D$ can be found in time $\ca O(kd^2)$.
\end{lemma}
\begin{proof}
Let $\{T_1, \ldots, T_m\}$ be the connected components of $T'$ in $G_{i-1}^R$ and let $D_j = T_j \cap D'$ for each $j \in [m]$.
Observe that $(T_j, D_j)$ is a basic $(i-1)$-profile if and only if $|T_j| \le k(d+1)$.
Hence, if $|T_j| \le k(d+1)$ for each $j \in [m]$, then $\{(T_1, D_1), \ldots, (T_m, D_m)\}$ is a $D'$-decomposition of $\pi$; condition~\ref{decomp-red} is satisfied because there are no red edges between distinct $T_j$ and $T_\ell$.
Let $j \in [m]$ and observe that since $T$ is connected in $G_i^R$, there is a vertex $u \in T_j$ such that either $uv \in R(G_i)$ or $u \in \{u_1, u_2\}$.
Since $v$ has red degree at most $d$ in $G_i$, we have $m \le d+2$ as required.

Now suppose that $|T_j| > k(d+1)$ for some $j \in [m]$.
Since $|T| \le k(d+1)$ and $T_j \seq T'$, we deduce that $|T_j| = k(d+1) +1$, $m = j = 1$, and $(T_j, D_j) = (T', D')$.
Since each vertex in $D'$ has red degree at most $d$ in $G_{i-1}$ and $|D'| \le k$, there is a vertex $w \in T' \setminus D'$ such that $wx \notin R(G_{i-1})$ for each $x \in D'$.
Let $\{T_1', \ldots, T_p'\}$ be the connected components of $T' \setminus \{w\}$ in $G_{i-1}^R$ and let $D_j' = T_j' \cap D'$ for each $j \in [p]$.
Now observe that $\{(T_1', D_1'), \ldots, (T_p', D_p')\} \cup \{(\{w\}, \emptyset)\}$ is a $D'$-decomposition of $\pi$.
Since $T_j'$ contains a vertex adjacent to $w$ in $G_{i-1}^R$ for each $j \in [p]$, we have $p+1 \le d+2$ as required.

Observe that the proof naturally translates into an algorithm finding the desired decomposition and that all the operations can be performed in time linear in $|T'| + |R(G_{i-1}[T'])|$.
Since $|T'| \le k(d+1) + 1$ and each vertex of $T'$ has red degree at most $d$ in $G_{i-1}$, there are $\ca O(kd^2)$ red edges in $G_{i-1}[T']$, and so the stated running time is achieved.
\end{proof}

Now we prove that the number of basic $i$-profiles containing any fixed vertex of $G_i$ is small. 

\begin{lemma}\label{lem:number-of-basic}
If $i \in [n]$ and $v \in V(G_i)$, then the number of basic $i$-profiles $(T, D)$ such that $v \in T$ is in $d^{\ca O(kd)}$.
Moreover, such profiles can be enumerated in time $d^{\ca O(kd)}$.
\end{lemma}
\begin{proof}
Let $(T, D)$ be a basic $i$-profile such that $v \in T$.
By Corollary 3.2. from \cite{tww3}, there are at most $d^{2\ell-2}+1$ red-connected subsets of $V(G_i)$ of size at most $\ell$ containing $v$. 
Since $|T| \le k(d+1)$, there are $d^{2kd+2k-2}+1$ possible choices of $T$. Since $D \seq T$, there are $2^{k(d+1)}$ possible choices of $D$. Hence, there are $d^{\ca O(kd)}$ basic $i$-profiles containing $v$.
The claim about enumeration follows from Corollary 3.2. as well.
\end{proof}

\section{Partial Dominating Set and Vertex Cover}\label{sec:problems}

In this section, we show that the following parameterized problems are fixed-parameter tractable.
When a graph $G$ is clear from context, we denote the closed neighborhood of $S \seq V(G)$ in $G$ by $N[S]$, i.e., $N[S] = S \cup \{u \in V(G) \sep u$ has a neighbor in $S\}$.

\decisionproblem{Partial Dominating Set using twin-width}{A graph $G$, a contraction sequence for $G$ of width $d$, integers $k$ and $t$}{Is there a set $S \seq V(G)$ of size $k$ such that $|N[S]| \ge t$?}{$k + d$}

\decisionproblem{Partial Vertex Cover using twin-width}{A graph $G$, a contraction sequence for $G$ of width $d$, integers $k$ and $t$}{Is there a set $S \seq V(G)$ of size $k$ such that
at least $t$ edges have an endpoint in $S$?}{$k + d$}

Henceforth, we refer to these problems simply as \textsc{Partial Dominating Set} and \textsc{Partial Vertex Cover}.

\subsection{Extended Profiles}\label{sub:extended-profiles}

For both of these problems, we will need a notion of a profile, obtained by extending basic profiles introduced in Section~\ref{sec:prelims}.
Let us fix a graph $G$, a contraction sequence $(G = G_1 ,\ldots,G_n)$ of width $d$, and an integer $k$.

\begin{definition}\label{def:extended-profile}
Let $i \in [n]$ and let $\pi = (T, D)$ be basic $i$-profile.
An \emph{extended $i$-profile} is a tuple $\pi' = (T, D, M, f)$ such that $M \seq T$ and $f\colon T \rightarrow [0, k]$ is a function such that $\sum_{u\in T} f(u) \le k$ and for each $u \in T$,  $f(u) \le |\beta(u)|$, and $D = \{u \in T \sep f(u) \ne 0\}$.
\end{definition}

Informally, the function $f$ specifies the distribution of the dominating (resp. covering) vertices across
$D$, and the set $M$ specifies which bags we want to dominate vertices in.
Note that for \textsc{Partial Vertex Cover}, $M$ will not be used, i.e., $M = \emptyset$ will always hold (this allows us to use the same definition of a profile for both problems).
The following definition formalizes this idea.

\begin{definition}\label{def:sol}
Let $i \in [n]$ and let $\pi = (T, D, M, f)$ be an extended $i$-profile. A set $S \seq \beta(T)$ is a \emph{solution} of $\pi$ if for every $u \in T$, $f(u) = |S \cap \beta(u)|$.
We define the \emph{dominating-set-value} of $S$ in $\pi$ as $|N[S] \cap \beta(M)|$
and the \emph{vertex-cover-value} of $S$ in $\pi$ as $|\{uv \in E(G) \sep u \in S, v \in \beta(T)\}|$.
\end{definition}

Note that each set $S \seq V(G)$ can be a solution of more than one extended $i$-profile for each $i \in [n]$.
Indeed, $S$ determines only $D$ and $f$, whereas there are more choices of $T$ and $M$ ($T$ can be any small red-connected set containing $D$).
However, it will always be clear from context which profile a solution belongs to, and so we will often simply say, e.g., \emph{the dominating-set-value} of $S$ (omitting ``in $\pi$'').

Let us give a quick informal overview of how the profiles will be used to solve \textsc{Partial Dominating Set}.
During the algorithm, we proceed through each time $i$ ranging from $1$ to $n$ and compute for each extended $i$-profile $\pi$ its optimal value; that is, the maximum value $t$ such that a solution of $\pi$ dominates $t$ vertices from $\beta(M)$.
To compute the optimal value, we need to decompose the profile: each decomposition is a small set of extended $(i-1)$-profiles, see Section~\ref{subsec:extended-decompositions} for a formal definition.
Crucially, each solution of $\pi$ is described by some decomposition, and the best solution described by a decomposition $\ca D = \{\pi_1, \ldots, \pi_m\}$ can be computed from the optimal solutions of the profiles in $\ca D$ (which were computed in the previous iteration).
Hence, the optimal solution of $\pi$ can be computed by trying all possible decompositions.
More precisely, we do not try \emph{all} decompositions but only one decomposition of each ``type''.
Finally, to determine how many vertices can be dominated by a set of size $k$ in $G$, it suffices to examine the optimal value of one particular $n$-profile $\pi$ (note that $G_n$ consists of a single vertex, and so it is easy to define $\pi$ explicitly).  
We get this value at the end of the run of Algorithm \ref{alg:dom-set}.

Let us now show that the number of extended $i$-profiles containing any fixed vertex of $G_i$ is bounded.

\begin{lemma}\label{lem:number-of-extended}
If $i \in [n]$ and $v \in V(G_i)$, then the number of extended $i$-profiles $(T, D, M, f)$ such that $v \in T$ is in $2^{\ca O(kd\log (kd))}$.
Moreover, such profiles can be enumerated in time $2^{\ca O(kd\log (kd))}$.
\end{lemma}
\begin{proof}
By Lemma~\ref{lem:number-of-basic}, there are $d^{\ca O(kd)}$ basic $i$-profiles $(T, D)$.
Recall that $|T| \le k(d+1)$.
Since $M \seq T$, there are $2^{k(d+1)}$ possible choices of $M$.
Since $f$ is a function of type $T \rightarrow [0, k]$, there are at most $(k+1)^{k(d+1)}$ possible choices of $f$.
Hence, there are $d^{\ca O(kd)} \cdot 2^{\ca O(kd)} \cdot k^{\ca O(kd)} = 2^{\ca O(kd\log (kd))}$ extended $i$-profiles containing $v$, and the claim about enumeration follows from Lemma~\ref{lem:number-of-basic}.
\end{proof}

\subsection{Extended Decompositions}\label{subsec:extended-decompositions}
Now we generalize the $D'$-decompositions defined in Section~\ref{sec:prelims} to extended profiles.

Let $i \in [2, n]$, let $v \in V(G_i)$ be the new vertex in $G_i$, let $u_1, u_2 \in V(G_{i-1})$ be the two vertices contracted into $v$, let $\pi_0 = (T, D)$ be a basic $i$-profile such that $v \in T$,
let $T' = T \setminus \{v\}\cup\{u_1, u_2\}$,
and let $\pi = (T, D, M, f)$ be an extended $i$-profile.
Let $\hat M = M \setminus \{v\} \cup \{u_1, u_2\}$ if $v \in M$, and $\hat{M} = M$ otherwise.
We say that a function $f'\colon T' \rightarrow [0, k]$ is \emph{compatible} with $\pi$ if $f(v) = f'(u_1) + f'(u_2)$, $f'(u_1) \le|\beta(u_1)|$, $f'(u_2) \le|\beta(u_2)|$, and for each $u \in T \setminus \{v\}$, $f(u) = f'(u)$.
Note that there is at least one function compatible with $\pi$ because $f(v) \le |\beta(v)|$.

Let us fix a function $f'$ compatible with $\pi$, let $D' = \{u \in T' \sep f'(u) \ne 0\}$, and let $M'$ be the set defined as follows. 
\begin{equation}\label{eq:M'}
M' = \{u \in \hat{M} \sep u \text{ has no black neighbor in } D' \text{ in } G_{i-1}\}.
\end{equation}

We say that a set of extended $(i-1)$-profiles $\ca D = \{(T_j, D_j, M_j, f_j) \sep j \in [m]\}$ is an \emph{$f'$-decomposition} of $\pi$ if $f' = \bigcup_{j \in [m]} f_j$, $D' = \bigcup_{j \in [m]} D_j$, $M' = \bigcup_{j \in [m]} M_j$, and $\{(T_1, D_1), \ldots, (T_m, D_m)\}$ is a $D'$-decomposition of $\pi_0$.

\begin{remark}
Equation~\ref{eq:M'} allows us to avoid double-counting in the algorithm for \textsc{Partial Dominating Set}.
Indeed, if $u \in \hat{M}$ has a black neighbor $w$ in $D'$, then all vertices in $\beta(u)$ are dominated by some vertex in $\beta(w)$.
Hence, if we added $u$ to some $M_j$, then vertices in $\beta(u)$ would be counted as dominated twice.
\end{remark}

Now we prove that a small $f'$-decomposition exists for each $f'$ compatible with $\pi$.

\begin{lemma}\label{lem:extended-decomposition}
If $i \in [2, n]$, $V(G_i) \setminus V(G_{i-1}) = \{v\}$, $V(G_{i-1}) \setminus V(G_i) = \{u_1, u_2\}$, $\pi = (T, D, M, f)$ is an extended $i$-profile such that $v \in T$, $T' = T \setminus \{v\}\cup\{u_1, u_2\}$, and $f'\colon T' \rightarrow [0, k]$ is a function compatible with $\pi$, then there is an $f'$-decomposition $\ca D$ of $\pi$ with cardinality at most $d+2$.
Moreover, $\ca D$ can be found in time $\ca O(k^2d^2)$.
\end{lemma}
\begin{proof}
Let $D' = \{u \in T' \sep f'(u) \ne 0\}$, let $\hat{M} = M \setminus \{v\} \cup \{u_1, u_2\}$ if $v \in M$, and $\hat{M} = M$ otherwise, and let $M' \seq T'$ be the set satisfying Equation~\ref{eq:M'}; note that $M'$ is determined by $D'$, which in turn depends on $f'$.
By Lemma~\ref{lem:basic-decomposition}, there is a $D'$-decomposition $\ca D_0 = \{(T_1, D_1), \ldots, (T_m, D_m)\}$ of the basic $i$-profile $(T, D)$ such that $m \le d+2$.
For each $j \in [m]$, let $M_j = M' \cap T_j$ and let $f_j$ be the restriction of $f'$ to $T_j$.
It can be easily observed that
$\{(T_j, D_j, M_j, f_j) \sep j \in [m]\}$ is the desired $f'$-decomposition of $\pi$.

By Lemma~\ref{lem:basic-decomposition}, $\ca D_0$ can be found in time $\ca O(kd^2)$.
Since $|\hat{M}| \le |T'| \le k(d+1)+1$ and $|D'| \le k$, there are $\ca O(k^2d)$ possible black edges one has to inspect when computing $M'$.
Once $M'$ is computed, deriving $M_j$ (and $f_j$) for each $j \in [m]$ is trivial, which implies the stated running time.
\end{proof}

The following lemma shows that a solution of a profile $\pi$ can be constructed by ``merging'' the solutions of profiles in a decomposition of $\pi$.

\begin{lemma}\label{lem:solution-decomposition}
If $i \in [2, n]$, $\pi$ is an extended $i$-profile, $\{\pi_j \sep j \in [m]\}$ is an $f'$-decomposition of $\pi$, and $S_j$ is a solution of $\pi_j$ for each $j \in [m]$, then $S := \bigcup_{j \in [m]} S_j$ is a solution of $\pi$.
\end{lemma}
\begin{proof}
Let $\pi = (T, D, M, f)$ and for each $j \in [m]$, let $\pi_j = (T_j, D_j, M_j, f_j)$.
By Definition~\ref{def:sol}, $S_j \seq \beta(T_j)$ for each $j \in [m]$.
Since $T_1, \ldots, T_m$ are pairwise disjoint, $S_1, \ldots, S_m$ are pairwise disjoint as well.
Observe that $S \seq \bigcup_{j \in [m]}\beta(T_j) = \beta(T)$ as required.

Now we need to show that $f(u) = |S\cap \beta(u)|$ for each $u \in T$.
Let $v \in V(G_i)$ be the new vertex in $G_i$ and let $u_1, u_2 \in V(G_{i-1})$ be the two vertices contracted into $v$.
First, let $u \in T \setminus \{v\}$.
Observe that there is a unique $j \in [m]$ such that $u \in T_j$ and $f_j(u) = f'(u) = f(u)$ because $f'$ is compatible with $\pi$.
Since $S_j$ is a solution of $\pi_j$, we have $f(u) = |S_j \cap \beta(u)| = |S\cap \beta(u)|$ as required.
Second, let $j, \ell \in [m]$ be such that $u_1 \in T_j$ and $u_2 \in T_\ell$.
Since $f'$ is compatible with $\pi$, we have $f(v) = f'(u_1) + f'(u_2)$.
Since $S_j$ (resp. $S_\ell$) is a solution of $\pi_j$ (resp. $\pi_\ell$), we obtain $f(v) = |S_j \cap \beta(u_1)| + |S_\ell \cap \beta(u_2)| = |S \cap (\beta(u_1) \cup \beta(u_2))| = |S \cap \beta(v)|$ as required.
\end{proof}

\subsection{Partial Dominating Set}\label{sub:ds}

Let $(G, C = (G_1, \ldots, G_n), k, t)$ be an instance of the \textsc{Partial Dominating Set} problem; recall that $C$ is a contraction sequence of width $d$.
In this section, we show that Algorithm~\ref{alg:dom-set} solves the problem in time $(dk)^{\ca O(dk)}n$, see Theorem~\ref{thm:dominating-set}.
For an outline of the algorithm, see Section~\ref{sub:extended-profiles}.

We will use the notion of extended profiles, see Definition~\ref{def:extended-profile}, and the notion of solutions, see Definition~\ref{def:sol}.
Instead of the dominating-set-value of a solution, we will say simply the \emph{value}.
We say that a solution $S$ is an \emph{optimal} solution of an extended profile $\pi$ if no other solution of $\pi$ has higher value than $S$.

Before presenting the algorithm, let us make two key observations.

\begin{lemma}\label{lem-ds:observations}
Let $i \in [2, n]$, $\pi$ be an extended $i$-profile, $\{(T_j, D_j, M_j, f_j) \sep j \in [m]\}$ be an $f'$-decomposition of $\pi$, $S$ be a solution of $\pi$, and $S_j = S \cap \beta(T_j)$ for each $j \in [m]$. For each $j \in [m]$:
\begin{enumerate}[a)]
\item The value of $S_j$ equals $|N[S] \cap \beta(M_j)|$.\label{dsitem:mj}
\item If $u \in T_j$ has a black neighbor in $D_\ell$ for some $\ell \in [m]$, then $\beta(u) \seq N[S]$.\label{dsitem:cj}
\end{enumerate}
\end{lemma}

\begin{proof}
For~\ref{dsitem:mj}), let us fix $j \in [m]$. It suffices to prove that $N[S_j] \cap \beta(M_j) = N[S] \cap \beta(M_j)$.
Suppose for contradiction that there are vertices $u_0 \in \beta(M_j)$ and $v_0 \in S \setminus S_j$ such that $u_0v_0 \in E(G)$.
Let $u, v \in V(G_{i-1})$ be such that $u_0 \in \beta(u)$ and $v_0 \in \beta(v)$, and observe that $v \in D_\ell$ for some $\ell \in [m]$, $\ell \ne j$.
Since $u \in M_j$, $uv$ is not a black edge in $G_{i-1}$, see Equation~(\ref{eq:M'}).
However, $uv$ is not a red edge in $G_{i-1}$ either, by condition~\ref{decomp-red} in the definition of a basic decomposition in Section~\ref{sec:prelims}.
This is a contradiction with $u_0v_0 \in E(G)$.

For~\ref{dsitem:cj}), let us fix $j \in [m]$ and $u \in T_j$ that has a black neighbor $v \in D_\ell$ for some $\ell \in [m]$.
By Definition~\ref{def:extended-profile}, $f'(v) \ne 0$ and there is a vertex $w \in \beta(v) \cap S$.
Hence, $\beta(u) \seq N(w) \seq N[S]$ as required.
\end{proof}

\newcommand{\com}[1]{\textup{\textbf{#1}}}
\newcommand{\Each}{\com{each}\xspace}

\begin{algorithm}[h]
\DontPrintSemicolon
\KwIn{A graph $G$, contraction sequence $(G_1 = G, \ldots, G_n)$ of width $d$, and $k, t \in \bb N$.}
\lFor{\Each extended $1$-profile $\pi = (T, D, M, f)$}{
$\sigma_1(\pi) := (D, |D\cap M|)$\label{dsl:one}
}
\For{$i$ \com{from} $2$ \com{to} $n$}{\label{dsl:for-i}
$v := $ the new vertex in $G_i$\;
$u_1, u_2 := $ the two vertices of $G_{i-1}$ contracted into $v$\;
\For{\Each extended $i$-profile $\pi = (T, D, M, f)$}{\label{dsl:for-profile}
	\lIf{$v \notin T$}{$\sigma_i(\pi) := \sigma_{i-1}(\pi)$; \textbf{continue} with next $\pi$\label{dsl:easy}}
	$\sigma_i (\pi) := (\textsf{null}, -1)$\;
	$T' := (T\setminus v)\cup \{u_1,u_2\}$\;
	\lIf{$v \in M$}{$\hat{M} := M \setminus \{v\} \cup\{u_1, u_2\}$ \textbf{else} $\hat{M} := M$}\label{dsl:mstar}
	\For{\Each function $f'\colon T' \rightarrow [0,k]$ compatible with $\pi$}{\label{dsl:for-function}
		Let $\{\pi_j \sep j \in [m]\}$ be an $f'$-decomposition of $\pi$ s.t. $m \le d+2$; it exists by Lemma~\ref{lem:extended-decomposition}.
		\;\label{dsl:decomposition}
		\For {\Each $j \in [m]$}{
			$(T_j, D_j, M_j, f_j) := \pi_j$;
			$(S_j, \nu_j) := \sigma_{i-1}(\pi_j)$\;\label{dsl:values}
			$C_j = \{u \in \hat{M} \cap T_j \sep u$ has a black neighbor in $D_\ell$ for some $\ell \in [m]\}$\;\label{dsl:cj}
			$\textsf{value}_j := \nu_j + \sum_{u \in C_j} |\beta(u)|$\;\label{dsl:valuej}
			}
		$\nu := \sum_{j\in [m]} \textsf{value}_j$; $S := \bigcup_{j \in [m]} S_j$\;\label{dsl:solution}
		\lIf{$\nu \ge \nu_0$, where $(S_0, \nu_0)= \sigma_i(\pi)$}{$\sigma_i(\pi) := (S, \nu)$
		}}
	}\label{dsl:end-for}
}
$(S, \nu) := \sigma_n(\{u\}, \{u\}, \{u\}, \{(u,k)\})$, where $\{u\} := V(G_n)$\;\label{dsl:output}
\Return $\nu \ge t$

\caption{Algorithm for \textsc{Partial Dominating Set}}\label{alg:dom-set}
\end{algorithm}

Now we show that in each iteration of the loop on lines~\ref{dsl:for-function}-\ref{dsl:end-for} in Algorithm~\ref{alg:dom-set}, a solution of the currently processed profile $\pi$ and its value are computed.
Informally, for each profile $\pi_j$ in the decomposition of $\pi$, we avoid double-counting dominated vertices in $\beta(\hat{M} \cap T_j)$ by partitioning $\hat{M} \cap T_j$ into two subsets, $M_j$ and $C_j$: $\beta(M_j)$ contains vertices dominated ``from inside'', i.e., via a red edge considered in one of the previous iterations, and $\beta(C_j)$ contains those dominated ``from outside'', i.e., via a black edge to $D'$.

\begin{lemma}\label{lem-ds:is-solution}
Let $i \in [2, n]$, $\pi$ be an extended $i$-profile, $f'$ be a function compatible with $\pi$, and $\nu$ and $S$ be the values computed by Algorithm~\ref{alg:dom-set} on line~\ref{dsl:solution} in the iteration processing $f'$.
If for each extended $(i-1)$-profile $\pi'$, $S'$ is a solution of $\pi'$ of value $\nu'$, where $(S', \nu')$ is the pair computed in $\sigma_{i-1}(\pi')$, then $S$ is a solution of $\pi$ of value $\nu$.
\end{lemma}
\begin{proof}
Let $\pi = (T, D, M, f)$ and let $\{\pi_j \sep j \in [m]\}$ be the $f'$-decomposition of $\pi$ found on line~\ref{dsl:decomposition} in the iteration processing $f'$.
For each $j \in [m]$, let $T_j, D_j, M_j, f_j, S_j, \nu_j$, and $C_j$ be the values computed on lines~\ref{dsl:values} and~\ref{dsl:cj}. By Lemma~\ref{lem:solution-decomposition}, $S$ is a solution of $\pi$.
To simplify notation, let $\delta(A) := |N[S] \cap \beta(A)|$ for any $A \seq V(G_i)$ or $A \seq V(G_{i-1})$.
Now it suffices to show that $\nu = \delta(M)$.

Let $\hat{M}$ be defined as in the algorithm, see line~\ref{dsl:mstar}, and observe that $\delta(M) = \delta(\hat{M})$.
Since $T_1, \ldots, T_m$ are pairwise disjoint, it suffices to show that for each $j \in [m]$, $\textsf{value}_j$ computed on line~\ref{dsl:valuej} equals $\delta(\hat{M} \cap T_j)$.
Let us fix $j \in [m]$, and observe that $\hat{M} \cap T_j = M_j \cup  C_j$ and $M_j \cap C_j = \emptyset$, see Equation~(\ref{eq:M'}) and line~\ref{dsl:cj}.
By Lemma~\ref{lem-ds:observations}\ref{dsitem:mj}, we have $\nu_j = \delta(M_j)$ because $\nu_j$ is the value of $S_j$.
By Lemma~\ref{lem-ds:observations}\ref{dsitem:cj}, $\beta(u) \seq N[S]$ for each $u \in C_j$,
which implies $\delta(C_j) = |\beta(C_j)|$.
Hence, $\textsf{value}_j = \delta(M_j) + \delta(C_j) = \delta(\hat{M} \cap T_j)$, and $\nu = \delta(M)$ as required.
\end{proof}

Now we show that we are able to find an \emph{optimal} solution for each profile.
The idea is to describe an optimal solution with a function $f'$ compatible with the considered profile, and then to show that the solution found in the iteration processing $f'$ is optimal as well. 

\begin{lemma}\label{lem-ds:optimal-solution}
For each $i \in [n]$ and each extended $i$-profile $\pi$,
if the final value of $\sigma_i(\pi)$ computed by Algorithm~\ref{alg:dom-set} is $(S, \nu)$, then $S$ is an optimal solution of $\pi$ and $\nu$ is the value of $S$.
\end{lemma}

\begin{proof}
Let $i \in [n]$ and $\pi = (T, D, M, f)$ be an extended $i$-profile.
We proceed by induction on $i$.
First, suppose that $i = 1$.
Since there are no red edges in $G_1 = G$ and $T$ is connected in $G^R$, we know that $T = \{u\}$ for some $u \in V(G)$.
Notice that $D, M \in \{\emptyset, \{u\}\}$ and $f \in \{\{(u,0)\},\{(u,1)\}\}$.
In particular, $\pi$ has only one solution, namely $D$, and the value of $D$
is $|D \cap M|$ as required, see line~\ref{dsl:one}.
Now suppose that $i > 1$.
Let $v \in V(G_i)$ be the new vertex in $G_i$ and let $u_1, u_2 \in V(G_{i-1})$ be the two vertices contracted into $v$.
If $v \notin T$, then $\pi$ is also an $(i-1)$-profile.
By induction hypothesis, $\sigma_{i-1}(\pi)$ is an optimal solution of $\pi$, which ensures that $\sigma_i(\pi)$ is also an optimal solution of $\pi$, see line~\ref{dsl:easy}.
Hence, let us assume that $v \in T$ and let $T' = (T\setminus v)\cup \{u_1,u_2\}$.
By induction hypothesis and Lemma~\ref{lem-ds:is-solution}, $\pi$ has a solution since there is a function compatible with $\pi$.

Let $S^*$ be an optimal solution of $\pi$, let $\nu^*$ be the value of $S^*$, and let $f'\colon T' \rightarrow [0, k]$ be defined as $f'(u) = |S^* \cap \beta(u)|$.
Observe that by Definition~\ref{def:sol}, $f'(u) = f(u)$ for each $u \in T \setminus \{v\}$, and $f(v) = |S^* \cap \beta(v)| = |S^* \cap \beta(u_1)| + |S^* \cap \beta(u_2)| = f'(u_1) + f'(u_2)$, i.e., $f'$ is compatible with $\pi$.
Let us consider the iteration of the loop on lines~\ref{dsl:for-profile}-\ref{dsl:end-for} processing $\pi$ and the iteration of the loop on lines~\ref{dsl:for-function}-\ref{dsl:end-for} processing $f'$.
Let $\{\pi_j \sep j \in [m]\}$ be the $f'$-decomposition of $\pi$ found on line~\ref{dsl:decomposition}.
For each $j \in [m]$, let $T_j, D_j, M_j, f_j, S_j, \nu_j$, and $C_j$ be as in the algorithm.
Let $\nu$ and $S$ be the values computed on line~\ref{dsl:solution}.
By Lemma~\ref{lem-ds:is-solution}, $S$ is a solution of $\pi$ with value~$\nu$.
Similarly to the proof of Lemma~\ref{lem-ds:is-solution}, we denote $\delta(A) = |N[S] \cap \beta(A)|$ and $\delta^*(A) = |N[S^*] \cap \beta(A)|$ for any $A \seq V(G_i)$ or $A \seq V(G_{i-1})$.

Suppose for contradiction that $\nu < \nu^*$, i.e., $\delta(M) < \delta^*(M)$, see Definition~\ref{def:sol}.
Let $\hat{M}$ be as in the algorithm, see line~\ref{dsl:mstar}, and observe that $\delta(M) = \delta(\hat{M})$ and $\delta^*(M) = \delta^*(\hat{M})$.
Hence, there must be $j \in [m]$ such that $\delta(\hat{M} \cap T_j) < \delta^*(\hat{M} \cap T_j)$; let us fix such a $j$.
By Lemma~\ref{lem-ds:observations}\ref{dsitem:cj}, if $u \in C_j$, then $\beta(u) \seq N[S] \cap N[S^*]$.
Hence, $\delta(C_j) = |\beta(C_j)| = \delta^*(C_j)$, which implies
$\delta(M_j) < \delta^*(M_j)$ because $\hat{M} \cap T_j = C_j \cup M_j$ and $M_j \cap C_j = \emptyset$.
Let $S_j^* = S^* \cap \beta(T_j)$ and observe that $S_j^*$ is a solution of $\pi_j$.
By Lemma~\ref{lem-ds:observations}\ref{dsitem:mj}, $\delta(M_j)$ is the value of $S_j$ and $\delta^*(M_j)$ is the value of $S_j^*$.
By induction hypothesis, $S_j$ is an optimal solution of $\pi_j$, which is a contradiction with $\delta(M_j) < \delta^*(M_j)$.
Hence, $\nu = \nu^*$, and $S$ is an optimal solution of $\pi$, too.

After the iteration of the loop on lines~\ref{dsl:for-function}-\ref{dsl:end-for} processing $f'$, $\sigma_i(\pi) = (S, \nu)$ because $S$ is an optimal solution of $\pi$.
Observe that the second component of $\sigma_i(\pi)$ cannot change in any of the following iterations because $\nu$ is always the value of some solution of $\pi$ by Lemma~\ref{lem-ds:is-solution}.
Hence, after all functions compatible with $\pi$ have been processed, $\sigma_i(\pi)$ is indeed an optimal solution of $\pi$ and its value, as desired.
\end{proof}

Using Lemma~\ref{lem-ds:optimal-solution}, it is easy to show the correctness of Algorithm~\ref{alg:dom-set}.

\begin{theorem}\label{thm:dominating-set}
Algorithm~\ref{alg:dom-set} decides the \textsc{Partial Dominating Set} problem in time $2^{\ca O(kd\log(kd))} \cdot n$.
\end{theorem}
\begin{proof}
Let $\iota = (G, C = (G = G_1, \ldots, G_n), k, t)$ be the input instance, where $C$ is contraction sequence of width $d$.
Let $V(G_n) = \{u\}$ and let $\pi = (\{u\}, \{u\}, \{u\}, \{(u,k)\})$.
Note that $\pi$ is an extended $n$-profile unless $n < k$, in which case, $\iota$ is trivially a NO instance.
By Lemma~\ref{lem-ds:optimal-solution}, $S$ is an optimal solution of $\pi$ with value $\nu$, where $(S, \nu)$ is the pair computed by Algorithm~\ref{alg:dom-set} on line~\ref{dsl:output}.
By Definition~\ref{def:sol}, $|S| = k$ and $\nu = |N[S] \cap \beta(\{u\})| = |N[S]|$.
Hence, $\iota$ is a YES instance if and only if $\nu \ge t$, as required.

Now need to argue that Algorithm~\ref{alg:dom-set} achieves the desired running time.
As an optimization, we maintain a single function $\sigma$ (instead of $\sigma_i$ for each $i\in [n]$).
Since there are $\ca O(n)$ extended $1$-profiles, line~\ref{dsl:one} can be performed in time $\ca O(n)$.
There are $n-1$ iterations of the outermost loop (i.e., the loop on lines~\ref{dsl:for-i}-\ref{dsl:end-for}); let us consider the iteration processing $i \in [2, n]$ and let $v$ be the new vertex in $G_i$.
Let $\pi = (T, D, M, f)$ be an extended $i$-profile.
Since $\sigma(\pi)$ has already been computed if $v \notin T$, it suffices to consider only profiles such that $v \in T$; these profiles can be enumerated in time $2^{\ca O(kd\log(kd))}$ by Lemma~\ref{lem:number-of-extended}.
There are at most $k+1$ functions $f'$ compatible with $\pi$ and for each of them, we can find an $f'$-decomposition of size at most $m \le d+2$ in time $\ca O(k^2d^2)$ by Lemma~\ref{lem:extended-decomposition}.
Hence, the running time of one iteration of the outermost loop is
$2^{\ca O(kd\log(kd))}$, which implies that the total running time is
$2^{\ca O(kd\log(kd))} \cdot n$ as required.
\end{proof}

\subsection{Partial Vertex Cover}\label{sub:vc}

In this section, we prove the following theorem.

\begin{theorem}\label{thm:vertex-cover}
There is an algorithm that, given a graph $G$, a contraction sequence $(G_1 = G, \ldots, G_n)$ of width $d$, and two integers $k$ and $t$, decides the \textsc{Partial Vertex Cover} problem in time $2^{\ca O(kd\log(kd))} \cdot n$.
\end{theorem}

The proof of the theorem above is similar to the proof of Theorem~\ref{thm:dominating-set}. 
Hence, we first discuss informally the differences between the two proofs and only afterwards give the formal proof in Section~\ref{sec:pvc-details}.
The first difference is already mentioned in Section~\ref{sub:extended-profiles};  
it suffices to consider extended $i$-profiles $\pi = (T, D, M, f)$ such that $M = \emptyset$.

Recall that the vertex-cover-value of a solution $S \seq \beta(T)$ is the number of edges in $G[\beta(T)]$ covered by $S$, see Definition~\ref{def:sol}.
Hence, if $\ca D = \{(T_j, D_j, \emptyset, f_j)\sep j \in [m]\}$ is a decomposition of $\pi$, then the edges in $G[\beta(T_j)]$ covered by $S$ are covered also by $S \cap \beta(T_j)$.
This means that when merging solutions of profiles in $\ca D$, we must only add the number $\delta_{j\ell}$ of covered edges with one endpoint in $\beta(T_j)$ and the other in $\beta(T_\ell)$ for \emph{distinct} $j,\ell \in [m]$.
Since there cannot be a red edge connecting $D_j$ and $T_\ell$ by definition of a basic decomposition, we can compute $\delta_{j\ell}$ by inspecting the \emph{black} edges between $T_j$ and $T_\ell$.
Observe that a black edge $uv \in B(G_{i-1})$ semi-induces a complete bipartite graph between $\beta(u)$ and $\beta(v)$ in $G$.
Because we have to avoid double-counting, we obtain the following equation:
\[
\delta_{j\ell} = 
\sum_{u \in T_j}\sum_{\substack{v \in T_\ell,\\ uv \in B(G_{i-1})}} f_j(u) \cdot |\beta(v)| + f_\ell(v) \cdot |\beta(u)| - f_j(u) \cdot f_\ell(v)
\]

Using this equation, we can compute an optimal solution of $\pi$.
The rest of the proof of Theorem~\ref{thm:vertex-cover} is analogous to the proof of Theorem~\ref{thm:dominating-set}.

\subsubsection{The Details for Partial Vertex Cover} \label{sec:pvc-details}

In this section, we show that Algorithm~\ref{alg:vertex-cover} solves the \textsc{Partial Vertex Cover} problem, thereby proving Theorem~\ref{thm:vertex-cover}.
We will reuse the notion of extended profiles, see Definition~\ref{def:extended-profile}, and the notion of solutions, see Definition~\ref{def:sol}.
However, we will not need the third component of a profile, and so it will always be set to $\emptyset$, i.e., a profile will be a tuple $(T, D, \emptyset, f)$.
Instead of the \emph{vertex-cover-value} of a solution, we will say simply the \emph{value}.
We say that a solution is an \emph{optimal} solution of $\pi$ if no other solution of $\pi$ has higher value than $S$.

\begin{algorithm}[h]
\DontPrintSemicolon
\KwIn{A graph $G$, contraction sequence $(G_1 = G, \ldots, G_n)$ of width $d$, and $k \in \bb N$.}
\lFor{\Each extended $1$-profile $\pi = (T, D, \emptyset, f)$}{
$\sigma_1(\pi) := (D, 0)$\label{vcl:one}
}
\For{$i$ \com{from} $2$ \com{to} $n$}{
$v := $ the new vertex in $G_i$\;
$u_1, u_2 := $ the two vertices of $G_{i-1}$ contracted into $v$\;
\For{\Each extended $i$-profile $\pi = (T, D, \emptyset, f)$}{\label{vcl:for-profile}
	\lIf{$v \notin T$}{$\sigma_i(\pi) := \sigma_{i-1}(\pi)$; \textbf{continue} with next $\pi$\label{vcl:easy}}
		$\sigma_i (\pi) := (\textsf{null}, -1)$\;
	$T' := (T\setminus v)\cup \{u_1,u_2\}$\;
	\For{\Each function $f'\colon T' \rightarrow [0,k]$ compatible with $\pi$}{\label{vcl:for-function}
		Let $\{\pi_j \sep j \in [m]\}$ be an $f'$-decomposition of $\pi$ s.t. $m \le d+2$; it exists by Lemma~\ref{lem:extended-decomposition}.\;\label{vcl:decomposition}
		\For {\Each $j \in [m]$}{
		$(T_j, D_j, \emptyset, f_j) := \pi_j$;
		$(S_j, \nu_j) := \sigma_{i-1}(\pi_j)$\;\label{vcl:values}
		}
		\For {\Each $1 \le j < \ell \le m$}{
            $\textsf{value}_{j\ell} := 0$\;
		    \For{\Each $x \in T_j$, $y \in T_\ell$ s.t. $xy \in B(G_{i-1})$\label{vcl:for-vertex}}{
                $\textsf{value}_{j\ell} := \textsf{value}_{j\ell} + f_j(x) \cdot |\beta(y)| + f_\ell(y) \cdot |\beta(x)| - f_j(x) \cdot f_\ell(y)$\;\label{vcl:count-edges}
		    }
		}
		$S := \bigcup_{j \in [m]} S_j$;
		$\nu := \sum_{j\in [m]} \nu_j + \sum_{1 \le j < \ell \le m} \textsf{value}_{j\ell}$\;\label{vcl:solution}
		\lIf{$\nu \ge \nu_0$, where $(S_0, \nu_0)= \sigma_i(\pi)$}{$\sigma_i(\pi) := (S, \nu)$
		}}
	}\label{vcl:end-for}
}
$(S, \nu) := \sigma_n(\{u\}, \{u\}, \{u\}, \{(u,k)\})$, where $\{u\} := V(G_n)$\;\label{vcl:output}
\Return $\nu \ge t$
\caption{Algorithm for \textsc{Partial Vertex Cover}}\label{alg:vertex-cover}
\end{algorithm}

We will proceed analogously to Section~\ref{sub:ds}. The following lemma is the counterpart of Lemma~\ref{lem-ds:is-solution}.

\begin{lemma}\label{lem-vc:is-solution}
Let $i \in [2, n]$, $\pi = (T, D, \emptyset, f)$ be an extended $i$-profile, $f'$ be a function compatible with $\pi$, and $\nu$ and $S$ be the values computed by Algorithm~\ref{alg:vertex-cover} on line~\ref{vcl:solution} in the iteration processing $f'$.
If for each extended $(i-1)$-profile $\pi'$, $S'$ is a solution of $\pi'$ of value $\nu'$, where $(S', \nu')$ is the pair computed in $\sigma_{i-1}(\pi')$, then $S$ is a solution of $\pi$ of value $\nu$.
\end{lemma}
\begin{proof}
Let $\{\pi_j \sep j \in [m]\}$ be the $f'$-decomposition of $\pi$ found on line~\ref{vcl:decomposition} in the iteration processing $f'$, and for each $j \in [m]$, let $T_j, D_j, f_j, S_j$, and $\nu_j$ be the values computed on line~\ref{vcl:values}. By Lemma~\ref{lem:solution-decomposition}, $S$ is a solution of $\pi$.
To simplify notation, we denote $\delta(A, B) = |\{uv \in E(G) \sep u \in \beta(A), v \in \beta(B), \{u, v\} \cap S \ne \emptyset\}|$ for any $A, B \seq V(G_i)$ or $A, B \seq V(G_{i-1})$.
Hence, by Definition~\ref{def:sol}, we need to show that $\nu = \delta(T, T)$.

Observe that $\delta(T, T) =  \sum_{1 \le j \le \ell \le m} \delta(T_j, T_\ell)$.
Since $\nu_j$ is the value of $S_j$ and $S \cap T_j = S_j$, $\nu_j = \delta(T_j, T_j)$ for any $j \in [m]$.
Let us fix $j,\ell \in [m]$ such that $j < \ell$.
Now it suffices to show that $\textsf{value}_{j\ell} = \delta(T_j, T_\ell)$ is true on line~\ref{vcl:solution}.
Observe that $\delta(T_j, T_\ell) = \sum_{u \in T_j, v \in T_\ell} \delta(\{u\}, \{v\})$.
Let us fix $u \in T_j$ and $v \in T_\ell$, and let $\delta(u,v) := \delta(\{u\}, \{v\})$.
Suppose that $\delta(u, v) > 0$, which means that $u \in D_j$ or $v \in D_\ell$ because $S \cap \beta(\{u, v\}) \ne \emptyset$, see Definition~\ref{def:extended-profile}.
Observe that $uv \notin R(G_{i-1})$ by condition~\ref{decomp-red} in the definition of a basic decomposition, see Section~\ref{sec:prelims}.
Since $\delta(u, v) > 0$, we obtain $uv \in B(G_{i-1})$.
Finally, observe that in the iteration of the loop on lines~\ref{vcl:for-vertex}-\ref{vcl:count-edges}, $\delta(u, v)$ is added to $\textsf{value}_{j\ell}$.
Indeed, there are all possible edges between $\beta(u)$ and $\beta(v)$ in $G$, $f_j(u) = |S \cap \beta(u)|$ and $f_\ell(v) = |S \cap \beta(v)|$ by Definition~\ref{def:sol}, and the last term avoids double-counting the edges with both endpoints in $S$.
We have shown that $\nu = \delta(T, T)$ as desired.
\end{proof}

The following lemma is the counterpart of
Lemma~\ref{lem-ds:optimal-solution} in the partial dominating set
algorithm.

\begin{lemma}\label{lem-vc:optimal-solution}
For each $i \in [n]$ and each extended $i$-profile $\pi$,
if the final value of $\sigma_i(\pi)$ computed by Algorithm~\ref{alg:vertex-cover} is $(S, \nu)$, then $S$ is an optimal solution of $\pi$ and $\nu$ is the value of $S$.
\end{lemma}
\begin{proof}
Let $i \in [n]$ and $\pi = (T, D, \emptyset, f)$ be an extended $i$-profile.
First, suppose that $i = 1$.
Since $T = \{u\}$ for some $u \in V(G)$, there are no edges in $G[\beta(T)] = G[T]$, and the value of the only solution $D$ of $\pi$ is clearly 0, see line~\ref{vcl:one}.
Now suppose that $i > 1$.
Let $v \in V(G_i)$ be the new vertex in $G_i$ and let $u_1, u_2 \in V(G_{i-1})$ be the two vertices contracted into $v$.
If $v \notin T$, then $\pi$ is also an $(i-1)$-profile and, by induction hypothesis, $\sigma_i(\pi) = \sigma_{i-1}(\pi)$ is an optimal solution of $\pi$, see line~\ref{vcl:easy}.
Hence, let us assume that $v \in T$ and let $T' = (T\setminus v)\cup \{u_1,u_2\}$.
By induction hypothesis and Lemma~\ref{lem-vc:is-solution}, $\pi$ has a solution since there is a function compatible with $\pi$.

Let $S^*$ be an optimal solution of $\pi$, let $\nu^*$ be the value of $S^*$, and let $f'\colon T' \rightarrow [0, k]$ be defined as $f'(u) = |S^* \cap \beta(u)|$.
As in the proof of Lemma~\ref{lem-ds:optimal-solution}, $f'$ is compatible with $\pi$.
Let us consider the iteration of the loop on lines~\ref{vcl:for-profile}-\ref{vcl:end-for} processing $\pi$ and the iteration of the loop on lines~\ref{vcl:for-function}-\ref{vcl:end-for} processing $f'$.
Let $\{\pi_j \sep j \in [m]\}$ be the $f'$-decomposition of $\pi$ found on line~\ref{vcl:decomposition}.
For each $j \in [m]$, let $T_j, D_j, f_j, S_j$ and $\nu_j$ be as in the algorithm.
Let $\nu$ and $S$ be the values computed on line~\ref{vcl:solution}.
By Lemma~\ref{lem-vc:is-solution}, $S$ is a solution of $\pi$ with value~$\nu$.
Similarly to the proof of Lemma~\ref{lem-vc:is-solution}, we denote $\delta(A, B) = |\{uv \in E(G) \sep u \in \beta(A), v \in \beta(B), \{u, v\} \cap S \ne \emptyset\}|$ and $\delta^*(A) =|\{uv \in E(G) \sep u \in \beta(A), v \in \beta(B), \{u, v\} \cap S^* \ne \emptyset\}|$ for any $A \seq V(G_i)$ or $A \seq V(G_{i-1})$.

Suppose for contradiction that $\nu < \nu^*$, i.e., $\delta(T, T) < \delta^*(T, T)$.
Observe that $\delta(T, T) = \sum_{1 \le j \le \ell \le m} \delta(T_j, T_\ell)$ and $\delta^*(T, T) = \sum_{1 \le j \le \ell \le m} \delta^*(T_j, T_\ell)$.
Hence, there are $j,\ell \in [m]$ such that $\delta(T_j, T_\ell) < \delta^*(T_j, T_\ell)$.
Observe that for each $j, \ell \in [m]$, $j \ne \ell$, we have \[\delta(T_j, T_\ell) = \sum_{\substack{ u \in T_j, v \in T_\ell,\\ uv \in B(G_{i-1})}} f_j(u) \cdot |\beta(v)| + f_\ell(v) \cdot |\beta(u)| - f_j(u) \cdot f_\ell(v) = \delta^*(T_j, T_\ell).\]
Indeed, if $u \in T_j, v \in T_\ell$, and $uv \in R(G_{i-1})$, then $f_j(u) = 0 = f_\ell(v)$ by condition~\ref{decomp-red} in the definition of a basic decomposition, see Section~\ref{sec:prelims}.
Hence, there is $j \in [m]$ such that $\delta(T_j, T_j) < \delta^*(T_j, T_j)$.
However, $\delta(T_j, T_j)$ is the value of $S_j$ and $\delta^*(T_j, T_j)$ is the value of $S^* \cap \beta(T_j)$, which is a solution of $\pi_j$.
This is a contradiction since $S_j$ is an optimal solution of $\pi_j$ by induction hypothesis.
Hence, $\nu = \nu^*$, and $S$ is an optimal solution of $\pi$, too.

After the iteration of the loop on lines~\ref{vcl:for-function}-\ref{vcl:end-for} processing $f'$, $\sigma_i(\pi) = (S, \nu)$.
Observe that the second component of $\sigma_i(\pi)$ cannot change in any of the following iterations because $\nu$ is always the value of some solution of $\pi$ by Lemma~\ref{lem-vc:is-solution}.
Hence, after all functions compatible with $\pi$ have been processed, $\sigma_i(\pi)$ is indeed an optimal solution of $\pi$ and its value, as desired.
\end{proof}

We are now ready to prove the main theorem of this section. 
\begin{proof}[Proof of Theorem~\ref{thm:vertex-cover}]
Let $\iota = (G, C = (G = G_1, \ldots, G_n), k, t)$ be the input instance, where $C$ is contraction sequence of width $d$.
Let $V(G_n) = \{u\}$ and let $\pi = (\{u\}, \{u\}, \emptyset, \{(u,k)\})$.
Note that $\pi$ is an extended $n$-profile unless $n < k$, in which case, $(G, k, t)$ is a NO instance.
By Lemma~\ref{lem-vc:optimal-solution}, $S$ is an optimal solution of $\pi$ with value $\nu$, where $(S, \nu)$ is the pair computed by Algorithm~\ref{alg:vertex-cover}.
By Definition~\ref{def:sol}, $|S| = k$ and $\nu = |\{uv \in E(G) \sep u \in S, v \in \beta(\{u\})\}| = |\{uv \in E(G) \sep u \in S\}|$.
Hence, $(G, k, t)$ is a YES instance if and only if $\nu \ge t$.

The fact that Algorithm~\ref{alg:vertex-cover} achieves the desired running time can be proven analogously as for Algorithm~\ref{alg:dom-set}, see the proof of Theorem~\ref{thm:dominating-set}.
\end{proof}

\newcommand{\logic}{($\exists^*\sum\#$QF)}

\section{\logic Model Checking}\label{sec:logic}

We now turn our attention to the framework mentioned in the introduction.
For a background in (counting) logic, we refer the readers to~\cite{KuskeS2017}.
We consider a generalization of a fragment of the counting logic FO(>0), namely problems expressible by formulas of the form $\exists x_1 \cdots \exists x_k \sum_{\alpha \in I} \# y \, \psi_\alpha(x_1, \dots, x_k, y) \geq t$ for some quantifier-free formulas $\psi_\alpha$, $\alpha \in I$, and $t \in \N$.
Such a formula is true in a graph $G$ if there are vertices $v_1, \ldots, v_k \in V(G)$ such that $\sum_{\alpha \in I} \# y \, \psi_\alpha(v_1, \dots v_k, y) \geq t$, where $\# y \, \psi_\alpha(v_1, \dots, v_k, y)$ is the number of vertices $w \in V(G)$ such that $\psi_\alpha(v_1, \dots, v_k, w)$ holds in $G$.

We are interested in giving an FPT algorithm on graph classes of bounded twin-width to solve the model-checking problem for this logic. 

\decisionproblem{\logic Model Checking using twin-width}
{A graph $G$, a contraction sequence for $G$ of width $d$, a finite set $I$, a quantifier-free formula $\psi_\alpha$  with $k+1$ free variables for each $\alpha \in I$, and $t \in \N$}
{Are there vertices $v_1,\dots,v_k \in V(G)$ such that $G \models \sum_{\alpha \in I} \# y \, \psi_\alpha(v_1 \dots v_k,y) \geq t$?}
{$k + d$\footnote{Note that we assume that each formula $\psi_\alpha$ is normalized, i.e., there is no shorter formula equivalent to $\psi_\alpha$. Without this assumption, the parameter would have to be $k + d + \sum_{\alpha \in I}|\psi_\alpha|$.}}

It is not hard to express \textsc{Partial Dominating Set} in this logic, see Section~\ref{sec:intro}.
Let us illustrate how \textsc{Partial Vertex Cover} can be expressed.

\begin{example}\label{example:pvc}
We can express \textsc{Partial Vertex Cover} with $I = [k]$ and the following formulas.
For $\alpha \in I$, let $\psi_\alpha(x_1,\dots,x_k,y) \equiv E(x_\alpha, y) \land \bigwedge_{\gamma \in [\alpha-1]} y \ne x_\gamma$.
Intuitively, $\psi_\alpha$ counts the number of edges covered by $x_\alpha$.
To avoid double-counting, the potential edge $x_{\alpha}x_\gamma$ is counted by $\psi_\alpha$ only if $\alpha < \gamma$ (and otherwise it is counted by $\psi_\gamma$).
\end{example}

\subsection{Virtual Profiles}
In Section~\ref{sec:problems}, we used extended profiles, which were obtained by ``extending'' basic profiles introduced in Section~\ref{sec:prelims}.
In our model checking algorithm, we will again ``extend'' the basic profiles in a certain way.
This time, however, these profiles will need to contain even more information than the extended profiles.
Intuitively, the reason behind this change is that the $k$ existential variables in the considered formula are not interchangeable, whereas the $k$ vertices in a partial dominating set (or vertex cover) are.
Hence, a part of a virtual profile is a function $f$, which describes where each existential variable should be realized. 

However, we need to be able to capture the fact that some variables may be realized in a part of $G$ not covered by the profile; this is done by adding a bounded number of ``virtual'' vertices and edges ($\ca V$ and $\ca E$).  Note that the virtual vertices are ``new objects''; e.g., $V(G) \cap\ca V= \emptyset$.
Moreover, we need to remember which atomic formulas ($x_a = x_b$, $E(x_a, x_b)$, and negations of these two) are satisfied by which existential variables: this information is encoded in a vertex-labeled graph $H$.
Note that if a virtual profile has a solution, then $H$ must be compatible with $f$ and $\ca E$.

\newcommand{\sol}{\mathbf s}

\begin{definition}\label{def:virtual-profiles}
Let $i \in [n]$ and let $(T, D)$ be a basic $i$-profile.
A \emph{virtual $i$-profile} is a tuple $\pi = (T, D, \ca V, \ca E, f, H)$, where:
\begin{itemize}
\item $\ca V$ is the set of at most $k$ \emph{virtual} vertices.
\item $\ca E \seq \{uv \sep u \in \ca V \cup T, v \in \ca V\}$ is the set of \emph{virtual} edges.
\item $f\colon [k] \rightarrow D \cup \ca V$ is a surjective function such that $|f^{-1}(u)| \le |\beta(u)|$ for each $u \in D$.
\item $H$ is a graph with at most $k$ vertices with vertex set labeled as $\{h_1,\ldots,h_k\}$.
\end{itemize}
The \emph{expanded graph} of $\pi$, denoted $G_\pi$, is the graph with vertex set $\beta(T) \cup \ca V$ and edge set $E(G[\beta(T)]) \cup \{uv \in \ca E \sep u, v \in \ca V\} \cup \{u_0v \sep uv \in \ca E, v \in\ca V, u \in T, u_0 \in \beta(u)\}$.
A tuple $\sol = (s_1, \ldots, s_k) \in V(G_\pi)^k$ is a \emph{solution} of $\pi$ if:
\begin{itemize}
\item for each $a \in [k]$, $s_a = f(a)$ if $f(a) \in \ca V$, and $s_a \in \beta(f(a))$ if $f(a) \in D$;
\item the function $s_a \mapsto h_a$ is an isomorphism from $G_\pi[\{s_1, \ldots, s_k\}]$ to $H$.
\end{itemize} 
The \emph{value} of $\sol$ is defined as $\sum_{\alpha \in I} |\{u \in \beta(T) \sep G_\pi \models \psi_\alpha(s_1, \ldots, s_k, u)\}|$.
A solution $\sol$ of $\pi$ is \emph{optimal} if no other solution of $\pi$ has higher value than $\sol$.
\end{definition}

\smallskip

The expanded graph $G_\pi$ illustrates the concept of virtual profiles: 
the subgraph of $G$ induced by the bags from $T$ is expanded with the virtual vertices.
The virtual vertices are uniformly connected to vertices from each bag in $T$ as dictated by $\ca E$. 
A solution $\sol$ of $\pi$ must then adhere to $\pi$ in the following way:
each $s_a$ in $\sol$ has to be the correct virtual vertex or a vertex from the correct bag as dictated by $f$;
and $\sol$ has to induce the desired subgraph $H$ in the expanded graph $G_\pi$.
When an optimal solution has been computed for each virtual profile, we can find the answer to the  problem by inspecting an
$n$-profile $\pi^*$ such that $\ca V$ and $\ca E$ are empty, $T = V(G_n)$ and $\beta(T) = V(G)$, and $f$ is the unique function from $[k]$ to $V(G_n)$.
In contrast to the case of \textsc{Partial Dominating Set}, the profile $\pi^*$ is not unique; one has to try all possible choices of $H$.

Now we show that the number of virtual profiles is bounded.

\begin{lemma}\label{lem:number-of-virtual}
If $i \in [n]$ and $v \in V(G_i)$, then the number of virtual $i$-profiles $(T, D, \ca V, \ca E, f, H)$ such that $v \in T$ is in $2^{\ca O(k^2d\log(d))}$.
Moreover, such profiles can be enumerated in time $2^{\ca O(k^2d\log(d))}$.
\end{lemma}
\begin{proof}
By Lemma~\ref{lem:number-of-basic}, there is $d^{\ca O(kd)}$ basic $i$-profiles $(T, D)$.
Let us bound the number of possible graphs $H$.
There are at most $k\cdot 2^{k^2}$ graphs on at most $k$ vertices and for each of them,
there are at most $k^k$ vertex-labelings with labels $\{h_1, \ldots, h_k\}$.
Hence, there are $2^{\ca O(k^2)}$ choices of $H$.
There are $k+1$ choices of $\ca V$ (it suffices to choose the size of $\ca V$)
and at most $2^{k \cdot (k+k(d+1))} \in 2^{\ca O(k^2d)}$ choices of $\ca E$.
Finally, there are at most $(k+k(d+1))^k \in 2^{\ca O(k \log(kd))}$ choices of $f$.
Hence, there are $2^{\ca O(k^2d\log(d))}$ virtual $i$-profiles containing $v$, and the claim about enumeration follows from Lemma~\ref{lem:number-of-basic}.
\end{proof}

\subsection{Virtual Decompositions}

To achieve our goal of finding optimal solutions of virtual profiles, we continue similarly as in the previous section about \textsc{Partial Dominating Set} and \textsc{Partial Vertex Cover}. 
We will again need a notion of a decomposition for our virtual profiles.
However, we must ensure that the profiles in the decomposition are ``compatible'' with each other with respect to the virtual vertices.
Informally,
the virtual edges must correspond to the black edges between the profiles in $G_{i-1}$.

Let $i \in [2, n]$, let $v \in V(G_i)$ be the new vertex in $G_i$, let $u_1, u_2 \in V(G_{i-1})$ be the two vertices contracted into $v$, let $\pi_0 = (T, D)$ be a basic $i$-profile such that $v \in T$, let $\pi = (T, D, \ca V, \ca E, f, H)$ be a virtual $i$-profile, let $T' = T \setminus \{v\}\cup\{u_1, u_2\}$, and let $\ca E' = \{uw \in \ca E \sep u,w \in T' \cup \ca V\} \cup \{u_1w, u_2w \sep vw \in \ca E\}$.
We say that a function $f'\colon [k] \rightarrow T' \cup \ca V$ is \emph{compatible} with $\pi$ if $f'(a) \in \{u_1, u_2\}$ when $f(a) = v$, and $f'(a) = f(a)$ otherwise.

\begin{figure}
    \centering
    \begin{subfigure}[t]{0.16\textwidth}
    \centering
\begin{tikzpicture}[line cap=round,line join=round,>=triangle 45,x=1.0cm,y=1.0cm]
\begin{scope}[every node/.style={draw, circle, minimum width=12pt, minimum height = 12pt, inner sep=1pt}]
\node (c) {$1$};
\node (d)[above of = c] {$3$};
\node (e)[above of = d] {$56$};
\node (f)[right of = d] {$4$};
\node[left of = d] (a) {$2$};
\end{scope}
\draw (a)--(c)--(f)--(e)--(d)--(c);
\end{tikzpicture}    \end{subfigure}
    \hfill
    \kern-8pt 
    \begin{subfigure}[t]{0.16\textwidth}
        \centering
\begin{tikzpicture}[line cap=round,line join=round,>=triangle 45,x=1.0cm,y=1.0cm]
\begin{scope}[every node/.style={draw, circle, minimum width=12pt, minimum height = 12pt, inner sep=1pt}]
\node (a)[star, star points = 6, fill = violet!30!white] {$1$};
\node (bd)[above of = a, minimum width = 15pt] {$234$};

\node (x) [above of = bd, opacity=0] {};

\node (c)[left of = x] {};
\node (e)[right of = x] {};
\node (f)[above of = x] {$56$};

\draw[color = violet, thick] (c)--(a)--(bd);
\draw (c)--(f);
\draw[color=red, thick] (c)--(bd)--(f)--(e)--(bd);
\end{scope}
\end{tikzpicture}    \end{subfigure}
    \hfill 
    \begin{subfigure}[t]{0.16\textwidth}
        \centering
\begin{tikzpicture}[line cap=round,line join=round,>=triangle 45,x=1.0cm,y=1.0cm]
\begin{scope}[every node/.style={draw, circle, minimum width=12pt, minimum height = 12pt, inner sep=1pt}]
\node (a)[star, star points = 6, fill = violet!30!white] {$1$};

\node(x) [above of=a, opacity = 0]{};
\node(y) [above of=x, opacity = 0]{};

\node (b)[left of = x, fill=green!30] {$2$};
\node (d)[right of = x, fill=orange!20] {$34$};
\node (c)[above of = b, fill=green!30] {};
\node (e)[above of = d, fill=orange!20] {};
\node (f)[above of = y, fill=orange!20] {$56$};

\draw (d)--(f)--(c);
\draw (b)--(e);
\draw[color = violet, thick] (b)--(a) --(c);
\draw[color = violet, thick] (a)--(d);
\draw[color=red, thick] (b)--(c);
\draw[color=red, thick] (d)--(e)--(f);
\end{scope}
\end{tikzpicture}\end{subfigure}
    \hfill
    \begin{subfigure}[t]{0.16\textwidth}
        \centering
\begin{tikzpicture}[line cap=round,line join=round,>=triangle 45,x=1.0cm,y=1.0cm]
\begin{scope}[every node/.style={draw, circle, minimum width=12pt, minimum height = 12pt, inner sep=1pt, fill = green!30}]
\node (a) {$2$};
\node (b)[above of = a] {};
\end{scope}

\begin{scope}[every node/.style={draw, star, star points = 6, minimum width=12pt, fill= violet!30!white, minimum height = 12pt, inner sep=0pt}]
\node (c)[right of = a] {$1$};
\node (d)[above of = c] {$3$};
\node (e)[above of = d] {$56$};
\node (f)[right of = d] {$4$};
\end{scope}

\draw [color=red, thick] (a)--(b);
\draw [color=violet, thick] (a)--(c)--(b)--(e)--(d)--(c)--(f)--(e);
\end{tikzpicture}\end{subfigure}
    \hfill
    \begin{subfigure}[t]{0.16\textwidth}
        \centering
        \begin{tikzpicture}[line cap=round,line join=round,>=triangle 45,x=1.0cm,y=1.0cm]

\begin{scope}[every node/.style={draw, circle, minimum width=12pt, minimum height = 12pt, inner sep=1pt, fill=orange!20}]
\node (a) {$34$};
\node (b)[above of = a] {};
\node (c)[left of = b] {$56$};
\end{scope}

\begin{scope}[every node/.style={draw, star, star points = 6, minimum width=12pt, fill= violet!30!white, minimum height = 12pt, inner sep=0pt}]
\node (d)[right of = a] {$1$};
\node (e)[above of = d] {$2$};
\end{scope}

\draw (a)--(c);
\draw [color=red, thick] (a)--(b)--(c);
\draw [color=violet, thick] (a)--(d)--(e)--(b);
\end{tikzpicture}
    \end{subfigure}
    
    \caption{\textbf{Leftmost:} A graph $H$ with vertex set $\{h_1, \ldots, h_6\}$. The circle representing $h_a$ contains the digit $a$. \textbf{Middle left:} A virtual $i$-profile $\pi = (T, D, \ca V, \ca E, f, H)$ for $k = 6$. Virtual vertices are represented by violet stars and virtual edges by violet lines. $T$ is the set of all non-virtual vertices. For each $a \in [6]$, the vertex $f(a)$ contains the digit $a$. $D$ is the set $f(\{2, 5\})$. \textbf{Middle:} The circles, and black and red edges represent $G_{i-1}[T']$, the violet star is the vertex in $\ca V$, and the violet lines represent $\ca E'$. The digits represent the function $f'$. The set $T'$ has two red-connected components: $T_1$ colored in green and $T_2$ colored in orange. \textbf{Middle right:} The profile $(T_1, D_1, \ca V_1, \ca E_1, f_1, H)$. Note that the facts that $f_1(5) = f_1(6)$ and $f_1(3)f_1(4) \notin \ca E_1$ follow from $H$; they cannot be deduced from $T$ or $T'$. \textbf{Rightmost:} The profile $(T_2, D_2, \ca V_2, \ca E_2, f_2, H)$.}
    \label{fig:virtual}
\end{figure}
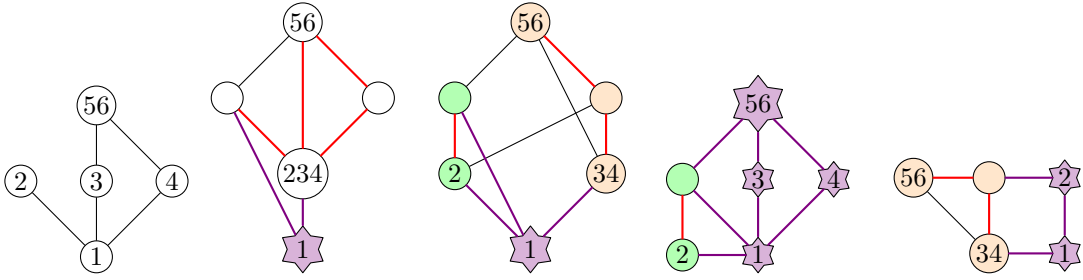

If a function $f'$ is compatible with $\pi$, then a set $\{\pi_j = (T_j, D_j, \ca V_j, \ca E_j, f_j, H) \sep j \in [m]\}$ of virtual $(i-1)$-profiles is an $f'$-\emph{decomposition} if $D' = \range(f') \cap T'$,
$\{(T_1,D_1),\ldots, (T_m, D_m)\}$ is a $D'$-decomposition of $\pi_0$, and for each $j \in [m]$:
\begin{itemize}
\item $\ca V_j = \ca V \cup \{h_a \sep f'(a) \in D' \setminus D_j\}$.
\item $f_j(a) = f'(a)$ if $f'(a) \in D_j \cup \ca V$, and $f_j(a) = h_a$ if $f'(a) \in D' \setminus D_j$.
\item $\ca E_j$ is the union of the following sets.
\begin{itemize}
\item $\{uw \in \ca E' \sep u,w \in \ca V \cup T_j\}$.
\item $\{uh_a \sep u \in \ca V \cup T_j, f'(a) = w \in D' \setminus D_j, uw \in \ca E' \cup B(G_{i-1})\}$.
\item $\{h_ah_b \in E(H) \sep f'(a), f'(b) \in D'\setminus D_j\}$.
\end{itemize}
\end{itemize}

The intuition behind this construction of $\pi_j$ is that for each vertex in $D' \setminus D_j$, which are the vertices with some variable assigned to them, we add a new virtual vertex to $\ca V_j$, and $\ca E_j$ contains $\ca E'$ ``restricted'' to $T_j$ (the first term), edges between a new virtual vertex $h_a$ and an ``old'' virtual vertex $u \in \ca V$ or a vertex $u \in T_j$ (the second term), and edges between two new virtual vertices (the third term).
Note that we used the vertices of $H$ to label the new virtual vertices, which means that two new virtual vertices $h_a$ and $h_b$ may be equal even when $a \ne b$; this ensures the possibility of an isomorphism from a solution of $\pi_j$ to $H$.
See Figure~\ref{fig:virtual} for an illustration of a virtual decomposition.

Again, we prove that a small $f'$-decomposition exists for each $f'$ compatible with $\pi$.

\begin{lemma}\label{lem:virtual-decomposition}
If $i \in [2, n]$, $V(G_i) \setminus V(G_{i-1}) = \{v\}$, $V(G_{i-1}) \setminus V(G_i) = \{u_1, u_2\}$, $\pi = (T, D, \ca V,\ca E, f, H)$ is a virtual $i$-profile such that $v \in T$, $T' = T \setminus \{v\}\cup\{u_1, u_2\}$, and $f'\colon [k] \rightarrow T' \cup \ca V$ is a function compatible with $\pi$, then there is an $f'$-decomposition $\ca D$ of $\pi$ with cardinality at most $d+2$.
Moreover, $\ca D$ can be found in time $\ca O(k^2d^2)$.
\end{lemma}
\begin{proof}
Let $D' = \range(f') \cap T'$.
By Lemma~\ref{lem:basic-decomposition}, a $D'$-decomposition $\ca D_0 = \{(T_j, D_j)\sep j\in[m]\}$ of $(T, D)$ such that $m \le d+2$ can be found in time $\ca O(kd^2)$.
For each $j \in [m]$, define $\ca V_j$, $\ca E_j$, and $f_j$ as in the definition of an $f'$-decomposition above.
Now it can be easily seen that $\{(T_j, D_j, \ca V_j, \ca E_j, f_j, H) \sep j \in [m]\}$ is the desired $f'$-decomposition of $\pi$.

Let $j \in [m]$ and observe that computing $\ca V_j$ and $f_j$ is straightforward.
To compute $\ca E_j$, we have to inspect all pairs $u \in T_j$ and $v \in D' \setminus D_j$.
The number of all such pairs, for all $j \in [m]$ in total, is bounded by $|T'| \cdot |D'| \in \ca O(k^2d)$.
After adding the time required to find $\ca D_0$, we obtain the stated running time.
\end{proof}

\subsection{Combining Solutions}

In Section~\ref{sec:problems}, we obtained a solution for an extended $i$-profile $\pi$ by simply taking the union of solutions of profiles in a decomposition of $\pi$.
In this section, combining solutions is more complicated because each solution is a tuple, which may contain some virtual vertices.

\begin{definition}\label{def:oplus}
Let $i \in [n]$, let $\pi_1 = (T_1, D_1, \ca V_1, \ca E_1, f_1, H)$ and $\pi_2 = (T_2, D_2, \ca V_2, \ca E_2, f_2, H)$ be two virtual $i$-profiles, and let $\sol^1 = (s_1^1,\ldots, s_k^1)$ and $\sol_2 = (s_1^2,\ldots, s_k^2)$ be solutions of $\pi_1$ and $\pi_2$, respectively.
By $\sol^1 \oplus \sol^2$, we denote the tuple $(s_1, \ldots, s_k) \in (\beta(D_1) \cup V(G_{\pi_2}))^k$ such that for each $a \in [k]$, $s_a = s_a^1$ if $f_1(a) \in D_1$, and $s_a = s_a^2$ otherwise.
\end{definition}

Note that in general, the operation $\oplus$ is neither associative nor commutative.
However, we will use it only to combine solutions of profiles in a decomposition (and when restricted to such inputs, $\oplus$ is both associative and commutative, which will be implicit in the proof of the following lemma).
Let $\sol = \bigoplus_{j \in [m]} \sol^j = \sol^1 \oplus (\sol^2 \oplus \ldots)$. Formally, $\sol = \sol^1$ if $m = 1$ and $\sol = \sol^1 \oplus \bigoplus_{j \in [2,m]} \sol^j$ otherwise.

\begin{lemma}\label{lem:oplus}
If $i \in [2,n]$, $\pi$ is a virtual $i$-profile, $\{\pi_j \sep j \in [m]\}$ is an $f'$-decomposition of $\pi$, and $\sol^j$ is a solution of $\pi_j$ for each $j \in [m]$, then $\sol = \bigoplus_{j \in [m]} \sol^j$ is a solution of $\pi$.
Moreover, if $\nu_j$ is the value of $\sol_j$, then the value of $\sol$ is $\sum_{j \in [m]} \nu_j$.
\end{lemma}
\begin{proof}
Let $\pi = (T, D, \ca V,\ca E, f, H)$ and $\sol = (s_1, \ldots, s_k)$.
For each $j \in [m]$, let $\pi_j = (T_j, D_j, \ca V_j, \ca E_j, f_j, H)$ and let $\sol_j = (s_1^j, \ldots, s_k^j)$.
For each $j \in [m]$, let $G^j := G_{\pi_j}$.

First, let $a \in [k]$ be such that $f(a) \in\ca V$.
Observe that for each $j \in [m]$, $f(a) = f'(a) = f_j(a) = s_a^j$, which means that $s_a = f(a)$ as desired (because $s_a = s_a^j$ for some $j \in [m]$ by Definition~\ref{def:oplus}).
Second, let $a \in [k]$ be such that $f(a) \in D$ and let $j \in [m]$ be the unique index such that $u := f'(a) = f_j(a) \in D_j$.
By Definition~\ref{def:oplus}, $s_a = s_a^j \in \beta(u)$ because $f_\ell(a) \notin D_\ell$ for each $\ell \in [m] \setminus \{j\}$.
By definition of $f'$, if $u \in V(G_{i-1}) \cap V(G_i)$, then $f(a) = u$.
Otherwise, $f(a)$ is the new vertex in $G_i$, and $s_a \in \beta(u) \seq \beta(f(a))$.
In both cases, $s_a \in \beta(f(a))$ as desired.

Now we need to show that the function $s_a \mapsto h_a$ is an isomorphism from $G_\pi[\{s_1, \ldots, s_k\}]$ to $H$.
Let us fix $a,b \in [k]$.
Since $\sol_j$ is a solution of $\pi_j$ for each $j \in [m]$, it suffices to show that $s_as_b \in E(G_\pi) \leftrightarrow s_a^js_b^j \in E(G^j)$ and $s_a = s_b \leftrightarrow s_a^j = s_b^j$ for some $j \in [m]$.
First, if $s_a, s_b \in \ca V$ or $s_a, s_b \in \beta(T_j)$ for some $j \in [m]$, then $s_a = s_a^j$ and $s_b = s_b^j$, and $s_as_b \in E(G_\pi) \leftrightarrow s_as_b \in E(G) \cup \ca E \leftrightarrow s_a^js_b^j \in E(G^j)$ as required.
Moreover, we clearly have $s_a = s_b \leftrightarrow s_a^j = s_b^j$; note that in all other cases, $s_a \ne s_b$ and $s_a^j \ne s_b^j$.
Second, if $s_a \in \ca V$ and $s_b \in \beta(u) \seq \beta(u')$ for some $u \in T_j$, $j \in [m]$, and $u' \in T$, then $s_as_b \in E(G_\pi) \leftrightarrow s_au' \in \ca E \leftrightarrow s_au \in \ca E_j \leftrightarrow s_as_b \in E(G^j)$.
Finally, assume that $j,\ell \in [m]$, $j \ne \ell$, and $s_a \in \beta(T_j), s_b \in \beta(T_\ell)$.
Let $u=f'(a)= f_j(a)$ and $v = f'(b) = f_\ell(b)$.
By definition of a solution, we have $s_a \in \beta(u)$ and $s_b \in \beta(v)$.
Recall that $D' = \bigcup_{p\in [m]}D_p$, that $u, v \in D'$, and that $\{(T_p, D_p)\sep p \in [m]\}$ is a $D'$-decomposition of the basic $i$-profile $(T, D)$.
Hence, by condition~\ref{decomp-red} in the definition of a $D'$-decomposition, see Section~\ref{sec:prelims}, $uv \notin R(G_{i-1})$.
Now $s_as_b \in E(G_\pi) \leftrightarrow s_as_b \in E(G) \leftrightarrow uv \in B(G_{i-1}) \leftrightarrow uh_b \in \ca E_j \leftrightarrow s_a^js_b^j \in E(G^j)$ as required.
Note that after the third equivalence, we could have used $h_av \in \ca E_\ell$ and continued analogously.
We have shown that $\sol$ is a solution of $\pi$.

What remains to be shown is that the value of $\sol$ equals $\sum_{j \in [m]} \nu_j$.
Since $\beta(T) = \beta(T')$, where $T' = \bigcup_{j \in [m]} T_j$, it suffices to show that for each $j \in [m]$, $u \in \beta(T_j)$ and $\alpha \in I$, $G_\pi \models \psi_\alpha(s_1, \ldots, s_k, u)$ if and only if $G^j \models \psi_\alpha(s_1^j, \ldots, s_k^j, u)$.
Using the isomorphisms to $H$, we obtain that $s_a\mapsto s_a^j$ is an isomorphism from $G_\pi[\{s_1, \ldots, s_k\}]$ to $G^j[\{s_1^j, \ldots, s_k^j\}]$.
Hence, it suffices to show that for each $a \in [k]$, (1) $u = s_a \leftrightarrow u = s_a^j$, and (2) $us_a \in E(G_\pi) \leftrightarrow us_a^j \in E(G^j)$.
If $s_a = s_a^j$, then (1) clearly holds, and if $s_a \ne s_a^j$, then $s_a \notin \beta(T_j)$ and $s_a^j \in \ca V_j$, which implies $u \notin \{s_a, s_a^j\}$.
Now let us focus on (2).
Let $u_j \in T_j$ and $u' \in T$ be such that $u \in \beta(u_j) \seq \beta(u')$.
If $s_a = s_a^j \in \beta(T_j)$, then $us_a \in E(G_\pi) \leftrightarrow us_a \in E(G) \leftrightarrow us_a \in E(G^j)$ as required.
If $s_a = s_a^j \notin \beta(T_j)$, then $s_a \in \ca V$, and $us_a \in E(G_\pi) \leftrightarrow u's_a \in \ca E \leftrightarrow u_js_a \in \ca E_j \leftrightarrow us_a \in E(G^j)$ as required.
Finally, suppose $s_a \ne s_a^j$, i.e., $s_a \in \beta(T) \setminus \beta(T_j)$, and $s_a^j = h_a \in \ca V_j$.
Let $w = f'(a)$ and observe that, by definition of a solution, $s_a \in \beta(w)$.
Since $w \in D'$, we have $u_jw \notin R(G_{i-1})$, see condition~\ref{decomp-red} in the definition of a $D'$-decomposition.
Hence, we have $us_a \in E(G_\pi) \leftrightarrow us_a \in E(G) \leftrightarrow u_jw \in B(G_{i-1}) \leftrightarrow u_jh_a \in \ca E_j \leftrightarrow uh_a = us_a^j \in E(G^j)$ as required.
We have proven that the value of $\sol$ is $\sum_{j \in [m]} \nu_j$ as desired.
\end{proof}

\subsection{Model Checking Algorithm}

\begin{algorithm}[h]
\DontPrintSemicolon
\KwIn{A graph $G$, a contraction sequence for $G$ of width $d$, a finite set $I$, a quantifier-free formula $\psi_\alpha$ with $k+1$ free variables for each $\alpha \in I$, and $t \in \N$}

\For{\Each virtual $1$-profile $\pi = (T, D, \ca V, \ca E, f, H)$}{\label{mcl:for-1-profile}
$\sol := (f(1), \ldots, f(k))$\;
\lIf{$\sol$ is a solution of $\pi$}{$\sigma_1(\pi) := (\sol, \sum_{\alpha \in I}|\{y \in T \sep G_\pi \models \psi_\alpha(\sol, y)\}|)$}\label{mcl:one}
\lElse{$\sigma_1(\pi) := (\textsf{null}, 0)$}
}\label{mcl:end-first-for}
\For{$i$ \com{from} $2$ \com{to} $n$}{\label{mcl:for-i}
$v := $ the new vertex in $G_i$;
$u_1, u_2 := $ the two vertices of $G_{i-1}$ contracted into $v$\;
\For{\Each virtual $i$-profile $\pi = (T, D, \ca V, \ca E, f, H)$}{\label{mcl:for-profile}
	\lIf{$v \notin T$}{$\sigma_i(\pi) := \sigma_{i-1}(\pi)$; \textbf{continue} with next $\pi$\label{mcl:easy}}
	$\sigma_i (\pi) := (\textsf{null}, 0)$\;
	$T' := (T\setminus v)\cup \{u_1,u_2\}$\;
	\For{\Each function $f'\colon [k] \rightarrow T' \cup \ca V$ compatible with $\pi$}{\label{mcl:for-function}
		Let $\{\pi_j \sep j \in [m]\}$ be an $f'$-decomposition of $\pi$ s.t. $m \le d+2$; it exists by Lemma~\ref{lem:virtual-decomposition}.\;\label{mcl:decomposition}
		\For{\Each $j \in[m]$}{$(\sol^j, \nu_j) := \sigma_{i-1}(\pi_j)$\;
		\lIf{$\sol^j = \textsf{null}$}{\textbf{continue} with next $f'$}\label{mcl:continue}}
		$\nu := \sum_{j\in [m]} \nu_j$; $\sol := \bigoplus_{j \in [m]} \sol^j$\;\label{mcl:solution}
		\lIf{$\nu \ge \nu_0$, where $\sigma_i(\pi) := (\sol_0, \nu_0)$}{$\sigma_i(\pi) := (\sol, \nu)$
		}}
	}\label{mcl:end-for}
}
$\textsf{solution} := \textsf{null}$; $\textsf{count} := 0$; $\{u\} := V(G_n)$\;
\For{each graph $H$ with vertex set labeled as $\{h_1, \ldots, h_k\}$}{\label{mcl:last-for}
$(\sol, \nu) := \sigma_n(\{u\}, \{u\}, \emptyset, \emptyset, f, H)$, where $f(a) = u$ for each $a \in [k]$\;\label{mcl:no-virtual}
\lIf{$\nu \ge \textsf{count}$}{$\textsf{solution} := \sol$; $\textsf{count} := \nu$}
}\label{mcl:end-last-for}
\Return $\textsf{count} \ge t$

\caption{Algorithm for \textsc{\logic Model Checking}}\label{alg:model-checking}
\end{algorithm}

Now we can start proving the correctness of Algorithm~\ref{alg:model-checking}.
First, let us show that the algorithm computes an optimal solution for each virtual profile $\pi$
(or it discovers that $\pi$ has no solution).

\begin{lemma}\label{lem:mc-optimal}
For each $i \in [n]$ and each virtual $i$-profile $\pi$,
if the final value of $\sigma_i(\pi)$ computed by Algorithm~\ref{alg:model-checking} is $(\sol, \nu)$, then $\sol$ is an optimal solution of $\pi$ and $\nu$ is the value of $\sol$.
On the other hand, if $\sigma_i(\pi) = (\textsf{null}, 0)$, then $\pi$ has no solution.
\end{lemma}
\begin{proof}
Let $i \in [n]$ and $\pi = (T, D, \ca V, \ca E, H)$ be a virtual $i$-profile.
First, suppose that $i = 1$.
Since $G = G_1$ has no red edges, $T = \{u\}$ for some $u \in V(G)$.
Recall that $\beta(u) = \{u\}$.
Hence, by Definition~\ref{def:virtual-profiles}, $(f(1), \ldots, f(k))$ is the only possible solution of $\pi$, and its value is either 0 or 1, see line~\ref{mcl:one}.
Now suppose that $i > 1$.
Let $v \in V(G_i)$ be the new vertex in $G_i$ and let $u_1, u_2 \in V(G_{i-1})$ be the two vertices contracted into $v$.
If $v \notin T$, then $\pi$ is also an $(i-1)$-profile and, by induction hypothesis, $\sigma_i(\pi) = \sigma_{i-1}(\pi)$ is an optimal solution of $\pi$ (or $(\textsf{null}, 0)$), see line~\ref{mcl:easy}.
Hence, let us assume that $v \in T$ and let $T' = (T\setminus v)\cup \{u_1,u_2\}$.

If $\pi$ has no solution, then for each function $f'$ compatible with $\pi$, the iteration of the loop on lines~\ref{mcl:for-function}-\ref{mcl:end-for} processing $f'$ must be exited on line~\ref{mcl:continue}; otherwise, a solution of $\pi$ would be found on line~\ref{mcl:solution} (by induction hypothesis and Lemma~\ref{lem:oplus}). Hence, $\sigma_i(\pi)$ always equals $(\textsf{null}, 0)$ as required.
Suppose that $\pi$ has a solution.
Let $\sol^* = (s_1^*, \ldots, s_k^*)$ be an optimal solution of $\pi$, let $\nu^*$ be the value of $S^*$, and let $f'\colon [k] \rightarrow T' \cup \ca V$ be such that $f'(a) = s_a^*$ if $s_a^* \in \ca V$, and $s_a^* \in \beta(f'(a))$ otherwise.
Observe that $f'$ is compatible with $\pi$, and let us consider the iteration of the loop on lines~\ref{mcl:for-profile}-\ref{mcl:end-for} processing $\pi$ and the iteration of the loop on lines~\ref{mcl:for-function}-\ref{mcl:end-for} processing $f'$.
Let $\{\pi_j \sep j \in [m]\}$ be the $f'$-decomposition of $\pi$ found on line~\ref{mcl:decomposition}.
For each $j \in [m]$, let $\pi_j = (T_j, D_j, \ca V_j,\ca E_j, H)$, let $\sol_j$ and $\nu_j$ be as in the algorithm, and let $\nu$ and $\sol$ be the values computed on line~\ref{mcl:solution}.
By Lemma~\ref{lem:oplus}, $\nu$ is the value of $\sol$.

For each $j \in [m]$, let $\sol^{j*} = (s_1^{j*}, \ldots, s_k^{j*})$ be the solution of $\pi_j$ such that for each $a \in [k]$, $s_a^{j*} = s_a^*$ if $s_a^* \in \beta(T_j) \cup\ca V$, and $s_a^{j*} = h_a$ otherwise.
Observe that $\sol^* = \bigoplus_{j\in[m]} \sol^{j*}$, and so, by Lemma~\ref{lem:oplus}, $\nu^* = \sum_{j\in[m]} \nu^*_j$, where $\nu^{*}_j$ is the value of $\sol^{j*}$.
Suppose for contradiction that $\nu < \nu^*$.
Since $\beta(T) = \beta(T')$ and $T' = \bigcup_{j \in [m]} T_j$, there is $j \in [m]$ such that $\nu_j <  \nu^{*}_j$, which is a contradiction since $\sol^j$ is an optimal solution of $\pi_j$ by the induction hypothesis.
Hence, $\sol$ is an optimal solution of $\pi$, too.

After the iteration of the loop on lines~\ref{mcl:for-function}-\ref{mcl:end-for} processing $f'$, $\sigma_i(\pi) = (\sol, \nu)$.
Observe that the second component of $\sigma_i(\pi)$ cannot change in any of the following iterations because $\nu$ is always the value of some solution of $\pi$ by Lemma~\ref{lem:oplus}.
Hence, after all functions compatible with $\pi$ have been processed, $\sigma_i(\pi)$ is indeed an optimal solution of $\pi$ and its value, as desired.
\end{proof}

Using Lemma~\ref{lem:mc-optimal}, it is easy to show the correctness of Algorithm~\ref{alg:model-checking}.

\begin{theorem}\label{thm:model-checking}
Algorithm~\ref{alg:model-checking} decides the \textsc{\logic Model Checking using twin-width} problem in time $2^{\ca O(k^2d\log(d))}n$.
\end{theorem}
\begin{proof}
Suppose that the algorithm is given a graph $G$, a contraction sequence $(G_1 = G, \ldots, G_n)$ of width $d$, and a formula $\phi \equiv \exists x_1 \cdots \exists x_k \sum_{\alpha \in I} \# y\, \psi_\alpha(x_1, \ldots, x_k, y) \ge t$.
Let us say that a graph $H$ with vertex set labeled as $\{h_1, \ldots, h_k\}$ is a \emph{template graph}.
Let us consider an iteration of the loop on lines~\ref{mcl:last-for}-\ref{mcl:end-last-for} processing a template graph $H$, 
let $\pi_H = (\{u\}, \{u\}, \emptyset, \emptyset, f, H)$ be as on line~\ref{mcl:no-virtual}, and let $(\sol_H, \nu_H)$ be the tuple computed on line~\ref{mcl:no-virtual}.
By Lemma~\ref{lem:mc-optimal}, $\sol_H$ is an optimal solution of $\pi_H$ with value $\nu_H$, or $(\sol_H, \nu_H) = (\textsf{null}, 0)$ if $\pi_H$ has no solution.
Since there are no virtual vertices in $\pi_H$, if $\pi_H$ has a solution, then $\sol_H$ is a tuple in $V(G)^k$, and $\nu_H = \sum_{\alpha \in I} |\{u \in V(G) \sep G\models \psi_\alpha(\sol, u)\}|$.

Let $(\sol, \nu)$ the final value of $(\textsf{solution, count)}$ computed by Algorithm~\ref{alg:model-checking}.
If $\sol = \textsf{null}$ and $\nu = 0$, then $\pi_H$ has no solution for any template graph $H$, which means that for any tuple $\sol' \in V(G)^k$, $\sum_{\alpha \in I}|\{u \in V(G) \sep G\models \psi_\alpha(\sol', u)\}| = 0$, and $(G, \phi)$ is a YES-instance if and only if $t = 0$.
Otherwise, $\sol \in V(G)^k$, $\nu = \sum_{\alpha \in I}|\{u \in V(G) \sep G\models \psi_\alpha(\sol, u)\}|$, and $\nu = \max\{\nu_H \sep H$ is a template graph$\}$.
Clearly, for each tuple $(s_1,\ldots, s_k) \in V(G)^k$, there is a template graph $H$ such that $s_a \mapsto h_a$ is an isomorphism from $G[\{s_1,\ldots, s_k\}]$ to $H$, which means that $(G, \phi)$ is a YES-instance if and only if $\nu \ge t$.
This shows the correctness of Algorithm~\ref{alg:model-checking}.

Now we need to argue that Algorithm~\ref{alg:model-checking} achieves the desired running time.
As in the proof of Theorem~\ref{thm:dominating-set}, we maintain a single function $\sigma$ (instead of $\sigma_i$ for each $i\in [n]$).
First, observe that by Lemma~\ref{lem:number-of-virtual}, there are at most $2^{\ca O(k^2d\log(d))}\cdot n$ virtual $1$-profiles. Hence, executing the for loop on lines~\ref{mcl:for-1-profile}-\ref{mcl:end-first-for} takes time at most $2^{\ca O(k^2d\log(d))}\cdot n$.
Second, observe that there are $n-1$ iterations of the loop on lines~\ref{mcl:for-i}-\ref{mcl:end-for}; let us consider the iteration processing $i \in [2, n]$ and let $v$ be the new vertex in $G_i$.
Let $\pi = (T, D, \ca V, \ca E, f, H)$ be a virtual $i$-profile.
Since $\sigma(\pi)$ has already been computed if $v \notin T$, it suffices to consider only profiles such that $v \in T$; these profiles can be enumerated in time $2^{\ca O(k^2d\log(d))}$ by Lemma~\ref{lem:number-of-virtual}.
There are at most $2^k$ functions $f'$ compatible with $\pi$ and for each of them, we can find an $f'$-decomposition of size at most $m \le d+2$ in time $\ca O(k^2d^2)$ by Lemma~\ref{lem:virtual-decomposition}.
Hence, the running time of one iteration of the loop is $2^{\ca O(kd\log(kd))}$.
Finally, since there are $2^{\ca O(k^2)}$ choices of $H$, executing the loop on lines~\ref{mcl:last-for}-\ref{mcl:end-last-for} takes time $2^{\ca O(k^2)}$.
In total, the running time is $2^{\ca O(k^2d\log(d))}\cdot n$ as required.
\end{proof}

\section{Conclusion}

We have demonstrated that several optimization problems requiring
counting can be efficiently solved on graphs with small twin-width.
This encompasses not only well-known problems such as {\sc Partial
Vertex Cover} and {\sc Partial Dominating Set}, but also all
problems expressible by a restricted fragment of first-order counting
logic. It would be desirable to extend this result by providing a more
comprehensive fragment that accommodates additional problems. For
bounded expansion, we can utilize formulas of the form $\exists
x_1\cdots\exists x_k\#y\,\phi(x_1,\ldots,x_k)$, where $\phi$ is an
arbitrary first-order formula rather than a quantifier-free one, and
even slightly more complex formulas~\cite{DreierR2021,MockR2024}. Can
we achieve the same for graphs of bounded twin-width? Moreover, we can
handle even more complicated formulas with nested counting quantifiers
\emph{approximately}~\cite{DreierR2021}. Once again, does the same or a
similar result hold true for bounded twin-width?  It should be noted
that the best result for nowhere-dense classes is basically the same
as our result for twin-width~\cite{DreierMR2023}.

\newpage

\bibliography{arxiv.bib}

\end{document}